\documentclass[12pt]{article}
\usepackage[margin=1in]{geometry}

\usepackage{style/qstyle}

\title{New Approaches to Complexity via Quantum Graphs}
\author[1]{Eric Culf}
\author[2]{Arthur Mehta}
\affil[1]{Faculty of Mathematics and Institute for Quantum Computing, University of Waterloo\footnote{\texttt{eculf@uwaterloo.ca}}}
\affil[2]{ Department of Mathematics and Statistics, University of Ottawa\footnote{\texttt{amehta2@uottawa.ca}}}
\date{\vspace{-10mm}}

\begin{document}

\pagenumbering{roman}
\setcounter{page}{1}

\maketitle

\begin{abstract}
Problems based on the structure of graphs --- for example finding cliques, independent sets, or colourings --- are of fundamental importance in classical complexity. Defining well-formulated decision problems for quantum graphs, which are an operator system generalisation of graphs, presents several technical challenges. Consequently, the connections between quantum graphs and complexity have been underexplored. 

In this work, we introduce and study the clique problem for quantum graphs. Our approach utilizes a well-known connection between quantum graphs and quantum channels. The inputs for our problems are presented as circuits inducing quantum channel, which implicitly determine a corresponding quantum graph. We show that, quantified over all channels, this problem is complete for $\tsf{QMA}(2)$; in fact, it remains $\tsf{QMA}(2)$-complete when restricted to channels that are probabilistic mixtures of entanglement-breaking and partial trace channels. Quantified over a subset of entanglement-breaking channels, this problem becomes $\tsf{QMA}$-complete, and restricting further to deterministic or classical noisy channels gives rise to complete problems for $\tsf{NP}$ and $\tsf{MA}$, respectively. In this way, we exhibit a classical complexity problem whose natural quantisation is $\tsf{QMA}(2)$, rather than $\tsf{QMA}$, and provide the first problem that allows for a direct comparison of the classes $\tsf{QMA}(2)$, $\tsf{QMA}$, $\tsf{MA}$, and $\tsf{NP}$ by quantifying over increasingly larger families of instances. 

We use methods that are inspired by self-testing to provide a direct proof of $\tsf{QMA}(2)$-completeness, rather than reducing to a previously-studied complete problem. We also give a new proof of the celebrated reduction of $\tsf{QMA}(k)$ to $\tsf{QMA}(2)$. In parallel, we study a version of the closely-related independent set problem for quantum graphs, and provide preliminary evidence that it may be in general weaker in complexity, contrasting to the classical case.

\end{abstract}

\newpage
\tableofcontents
\newpage
\pagenumbering{arabic}
\setcounter{page}{1}

\section{Introduction}

Shannon's foundational work on zero-error capacity established a deep connection between noisy channels and graph theory \cite{Sha56}. Shannon showed that to each classical noisy channel $N$, one can associate a corresponding graph $G_N$, called a \emph{confusability graph}, which encodes many important properties of the channel. Notably, $G_N$ will have an independent set (\emph{cf.} a clique) of size $k$ if and only if there is a collection of distinct inputs $x_1, \dots x_k$ for $N$ which have a zero (\emph{cf.} non-zero) probability of being confused as outputs.

A \emph{quantum graph} is a generalization of a graph which allows to extend the graph notion of confusability to quantum channels~\cite{DSW13}. In brief, given a quantum channel $\Phi$ there exists a corresponding \emph{operator system} $\mathcal{S}_{\Phi}$. 

A collection of states $\rho_1, \dots \rho_k$ forms an \emph{independent set (\emph{cf.} clique)} for $\mathcal{S}_{\Phi}$ if the average overlap $\Tr(\mathcal{N}(\rho_i) \mathcal{N}(\rho_j))$ is sufficiently small (\emph{cf.} large). Quantum graphs have been studied from various points of view including channel capacity, non-local games, and non-commutative relations~\cite{Sta16,TT20,KKS20,CW22,Mat22,CS22,Gan21,BEVW22}.

Due to the ubiquity of graphs in classical complexity, one might ask if it is also natural to study quantum graphs in the context of quantum complexity theory. In the classical setting, deciding if a graph has a clique (or an independent set) is $\tsf{NP}$-complete~\cite{Sip12}. The complexity class $\tsf{QMA}$ is recognized as a quantum analogue of $\tsf{NP}$, as the classical proof strings of $\tsf{NP}$ are replaced by general quantum states for this class. Kitaev emphasized this view by demonstrating $\tsf{QMA}$-completeness of the local Hamiltonian problem, a quantum analogue of the $\tsf{NP}$-complete problem $k$-SAT~\cite{KSV02}. Similarly, other $\tsf{NP}$-complete problems have been shown to quantise to $\tsf{QMA}$-complete problems, \emph{e.g.}~\cite{Bra06,BG22}. It is therefore reasonable to expect the clique and independent set problems for quantum graphs to be $\tsf{QMA}$-complete. 

Surprisingly, the complexity of these problems has received relatively little study. An initial treatment was given by Beigi and Shor~\cite{BS08}; however, as we discuss in \cref{subSec:BG08}, their work would require the additional assumption that the input states are product states in order to apply to our setting. An alternative quantum analogue of $\tsf{NP}$ is class $\tsf{QMA(2)}$, which contains $\tsf{QMA}$ and differs in that it is assumed that the that proofs are given as product states~\cite{KMY03,ABDFS09}. There are relatively few known $\tsf{QMA(2)}$-complete problems compared to the cornucopia of $\tsf{QMA}$-complete problems; moreover, known $\tsf{QMA(2)}$-complete problems like the sparse separable Hamiltonian  problem \cite{CS12}, and the product isometry output  problem \cite{GHMW15}, don't appear to be straightforward generalizations of $\tsf{NP}$-complete problems.

In this work, we study the decision problems of identifying cliques and independent sets in quantum graphs. We present our problem instances as circuit descriptions of quantum channels $\Phi$, which in turn determines a quantum graphs (operator systems) $\mathcal{S}_{\Phi}$. We also consider versions of the clique/independent set problems for graphs, where instead of an explicit description of a graph, the verifier is given a circuit description of a noisy channel which implicitly determines a graph, given as its confusability graph.

\vspace{-5mm}

\paragraph{Summary of Results} We examine the clique and independent set problems for quantum channels, and we show they are both in $\tsf{QMA}(2)$. Furthermore we prove that the clique problem is $\tsf{QMA}(2)$-complete and additionally show that, by quantifying over a restricted class of channels, we recover a complete problem for $\tsf{QMA}$. For deterministic and classical noisy channels, we show that our versions of both the clique and independent set problems are complete for $\tsf{NP}$ and $\tsf{MA}$, respectively. Our work outlines a hierarchy of decision problems which captures completeness for all four of these important complexity classes and provides evidence for viewing $\tsf{QMA(2)}$ as a natural quantum analogue of $\tsf{NP}$.

\subsection{Background}

In this section, we provide a brief overview of quantum graphs. For completeness we highlight that there are several overlapping notions of quantum graphs, which have been discovered in somewhat distinct settings, for example in early works due to Erd\H{o}s, Katavolos, Shulman \cite{EKS98}, and Weaver \cite{Wea10,Wea21}. More comprehensive introductions to the topics of operator systems and quantum graphs can be found in \cite{Pau02}, \cite{Pau16}, and \cite{Daw22}. 

We focus on the motivation of quantum graphs from the point of view of confusability graphs and describe how to include ordinary graphs\footnote{It is sometimes useful to emphasize the distinction between a graph (\emph{i.e.} vertex and edge set) and a quantum graph (\emph{i.e.} operator system). In this work, we highlight this by using the terms ordinary or classical graph.} as quantum graphs. We also discuss the technical challenges of giving a quantum graph as an input for a decision problem. Next, we outline our approach of presenting them as the description of a quantum circuit and then sketch the definitions of the $2$-clique and $2$-independent set problems. The extension of these definitions to larger sets of states is given in \cref{sec:motivation}. Lastly, we discuss how presenting problem instances as circuits descriptions of channels, which implicitly determines a confusability graph, motivates new versions of the clique and independent set problems for classical graphs as well.

\subsubsection{Quantum graphs}\label{subsec:Qgraphs} 

\paragraph{Confusability Graphs.} A classical noisy channel $N: X \rightarrow Y$ may be parametrised by a probability transition matrix $N=N(y|x)$, which gives the probability of receiving output $y$ given input $y$. Given a channel $N$, the corresponding \emph{confusability graph} $G_N$, is defined on vertex set $X$ with edges $x \sim x'$ whenever $N(y|x)N(y|x') \neq 0$, for some $y \in Y$, that is whenever there is non-zero probability that two inputs are mapped to the same output~\cite{Sha56}. A set $X' \subseteq X$ forms an independent set of $G_N$ if $x\not\sim x'$ for all distinct $x,x'\in X'$, which is if and only if the probability of any pair of distinct $x,x' \in X'$ being sent to the same output is $0$. Similarly, $X'$ forms a clique of $G_N$ if $x\sim x'$ for all distinct $x,x'\in X'$, which is if and only if any pair of distinct elements $x,x'\in X'$ have a non-zero probability as being received as the same output. 

In the quantum setting, classical noisy channels are superseded by quantum channels, represented by completely positive trace-preserving (CPTP) maps $\Phi: \mathbb{M}_n(\C) \rightarrow \mathbb{M}_m(\C)$. To each such quantum channel, one can associate an operator system $\mc{S}_{\Phi} \subseteq\mathbb{M}_n$, which is constructed as the linear span of all products  $A^*_iA_j$ for some set of Kraus operators $\lbrace A_i \rbrace_i$ of $\Phi$. The operator system $\mathcal{S}_{\Phi}$ is independent of the choice of Kraus representation and can be used to generalize the role of confusability graphs to study zero-error capacities of quantum channels. In particular, a collection of orthonormal states $\ket{v_1}, \dots \ket{v_k}$ can be distinguished from one another with certainty after passing through the channel if and only if $\ketbra{v_i}{v_j}$ is orthogonal to $\mathcal{S}_{\Phi}$, in which case we say $\ket{v_1}, \dots, \ket{v_k}$ is an \emph{(exact) independent set} for $\Phi$ or $\mathcal{S}_{\Phi}$~\cite{DSW13}. More generally, one can extend other notions from classical graph theory such as cliques and graph colourings to the setting of operator systems~\cite{OP16,Sta16,KM19,BTW21}. In this context, operator systems are often referred to as quantum graphs or non-commuative graphs. 

\paragraph{Classical Graphs as Quantum Graphs}
It is possible to embed classical graphs into the set of quantum graphs in the following way. Starting with an undirected graph $G$, define an operator system $\mathcal{S}_G$ by taking the span of $\ketbra{i}{j}$ over vertices $i$ and $j$ in $G$ satisfying $i=j$ or $i \sim_G j$. Graphs $G$ and $H$ are isomorphic if and only if $\mathcal{S}_G$ and $\mathcal{S}_H$ are unitarily equivalent~\cite{OP16}. Similarly, given a noisy channel $N=N(y|x)$, there are several ways to construct a corresponding quantum channel $\Phi_N$ which satisfies $\mathcal{S}_{\Phi_{N}} = \mathcal{S}_{G_{N}}$. Generalizations of cliques, independent sets, and colourings to quantum graphs respect this construction, \emph{i.e.} there is a correspondence between $k$-cliques, $k$-independent sets, and $k$-colourings for $G$ and those for the quantum graph $\mc{S}_G$~\cite{Pau16}.

\paragraph{The Complement of a Quantum Graph.} 
Recall that, given a graph $G$, the complement $\overline{G}$ is a graph on the same vertex set where the adjacency relation is defined as $i \sim_{\overline{G}} j$ if and only if $i\not\sim_{G} j$. However, there is no generally agreed-upon approach to extending the notion of the graph complement to quantum graphs. One straightforward approach would be to take the orthogonal complement of $\mathcal{S}$ inside $\mathbb{M}_n$ with respect to the Hilbert-Schmidt inner product, but this will not be an operator system as it fails to contain the identity element, and hence fail to correspond to any channel. One can remedy this by making a correction such as $\overline{\mc{S}} = \mathcal{S} + \C\Id$ to account for the identity~\cite{Sta16}. Unfortunately this does not extend the graph complement as $\mathcal{S}_{\overline{G}} \neq \overline{\mathcal{S}}_G$. Another approach allows both operator systems as well as their orthogonal complements to be considered as quantum graphs~\cite{KM19}. 

Relevant for this work is the fact there that there is no straightforward and unique notion of a graph complement for quantum graphs, and consequently the complexity of decision problems based on cliques and independent sets for quantum graphs cannot be reduced to one another by taking a complement. Indeed, as we discuss in \cref{subSec:OpenProblems}, it is even possible that they have different complexity.

\paragraph{Quantum Independent Sets and Quantum Cliques.} Motivated by the study of non-local games, quantum generalizations of many graph parameters have been found. The existence of $k$-cliques, $k$-independent sets, and $k$-colourings for a classical graph $G$ can be determined by the existence of a perfect classical strategy for a corresponding non-local game. In some cases, \emph{quantum strategies} for these games can outperform classical strategies, and perfect quantum strategies give rise to \emph{quantum cliques} and \emph{quantum independent sets} for classical graphs~\cite{MR16,HMPS19}. By considering non-local games where the questions are quantum states, one can also define quantum cliques and quantum independent sets for quantum graphs \cite{BGH19}.

It is known that determining the quantum value --- or commuting operator value --- of a non-local game is an undecidable problem \cite{CS19,JNV+20arxiv,MSY22}. These results can be captured using the non-local games associated to the graph parameters mentioned above, and thus it undecidable to determine the existence of quantum cliques/quantum independent sets in either classical and quantum graphs~\cite{Har2023}. In this work, we take a different perspective and do not consider quantum cliques or quantum independence sets. Instead, we exclusively study the complexity of determining the existence of cliques and independent sets in graphs and quantum graphs.

\subsubsection{The clique/independent set problems for quantum graphs}\label{subsec:PresentingGraphs-Channels}

Problems based on finding cliques, independent sets, or graph colourings are well-studied in classical complexity; they all give rise to $\tsf{NP}$-complete problems~\cite{Sip12}. In these cases, the input to the problem includes an explicit description of the vertices and edges of a graph. Providing an explicit description of a quantum graph raises some challenges: notably quantum graphs are linear subspaces $\mathcal{S} \subseteq \mathbb{M}_n$ and are thus determined by a basis of operators. As such, given a quantum channel $\Phi$ on an $n$-qubit system, the corresponding quantum graph $S_{\Phi}$ is a subspace with up to $2^{2n}$ dimensions, so specifying it in this way requires a description of an exponentially sized basis. Note that this problem also arises in the classical case, if the graphs in question are seen as arising from an noisy channel on an $n$-bit input space. For our work, we will instead consider quantum graphs as being determined by an underlying quantum channel $\Phi$, and give a problem instance as a quantum circuit $C$ describing the action of the channel $\Phi=\Phi_C$.

\paragraph{Presenting Quantum Graphs as Quantum Circuits.}
Different noisy channels can give rise to the same confusability graph, \emph{i.e.} the construction $N \mapsto G_N$ is not injective.
Similarly, in the quantum setting, the construction $\Phi \mapsto \mathcal{S}_\Phi$ is not injective. On the other hand, given any graph $G$ one can construct $N$ satisfying $G=G_N$ with an appropriate choice of transition probabilities $N=N(y|x)$; likewise given any operator system $\mathcal{S} \subseteq \mathbb{M}_n$, it is possible to define $\Phi$ such that $\mathcal{S}_{\Phi}= \mathcal{S}$~\cite{Pau16}. Consequently, quantifying over all quantum channels is sufficient to quantify over all quantum graphs, and it is hence well-motivated to consider parameters for quantum graphs at the level of channels. 

In \cref{Sec:QMA(2)}, we define approximate $k$-cliques and approximate $k$-independent sets at the level of quantum channels. For convenience we restate these definitions here. To keep the presentation concise, we only include the case $k=2$. In fact, restricting $k=2$ is sufficient for all of our complexity results in this paper, and in \cref{sec:a-k-cliques} we describe an explicit reduction of the corresponding decision problems from $k>2$ to $k=2$.

\begin{definition}\label{def:qClq/Ind:k=2}
    Let $\alpha \in [0,1]$. We say a pair of orthogonal states $\rho_1, \rho_2$ are an \emph{$(\alpha,2)$-clique} for $\Phi$ if $\Tr\squ*{\Phi(\rho_1)\Phi(\rho_2)}\geq\alpha$, and an \emph{$(\alpha,2)$-independent set}   if, $\Tr\squ*{\Phi(\rho_1)\Phi(\rho_2)}\leq 1 - \alpha$.  
\end{definition}

We remark on a few details motivating \cref{def:qClq/Ind:k=2}. Firstly, we require the inputs state to be exactly orthogonal. This agrees well with the classical noisy channel setting, which considers sets of distinct inputs $x_1 \neq x_2$. Second, we express the probability that the outputs of the channel are confused by means of the overlap under the Hilbert-Schmidt inner product $\Tr\squ*{\Phi(\rho_1)\Phi(\rho_2)}$. This also agrees well with classical notions of confusability since, if the output states are classical, \emph{i.e.} $\Phi(\rho_1)=\sum_x\lambda_x\ketbra{x}$ and $\Phi(\rho_2)=\sum_x\mu_x\ketbra{x}$, then the overlap $\sum_x\lambda_x\mu_x$ becomes simply the probability of the channels providing the same input on independent evaluations. Further, using the overlap as a measure of confusability is also natural in the computational setting, as the overlap can be efficiently sampled from the output states by means of the swap test, unlike other measures of closeness such as the trace distance. We note that we do not require the input states to be pure, though we can always assume the states in a clique or independent set are pure by convexity. Neither do we assume the output states are pure, as we quantify over general, not necessarily unitary, channels. Lastly, the definition reduces to the exact case for particular values of the parameters: $(1,k)$-independent sets of $\Phi$ are $k$-independent sets of $\mc{S}_\Phi$, and $(1,k)$-cliques are $k$-independent sets of $\mc{S}_\Phi^\perp$. 

What may be seen as a flaw of our definition is that it is possible for a mixed state $\rho$ to have small overlap with itself. This observation may seem somewhat contradictory to our measure of confusability, as this means that two inputs that have the same output state may nevertheless be close to independent, or \emph{pseudo-independent}. However, by seeing mixed states as samples from an ensemble of pure states, what this really means is that, upon taking two independent samples from the ensemble, they are unlikely to be equal. As we discuss in \cref{subsec:ProofTechniques}, this property actually plays an important role in our proof approach establishing $\tsf{QMA}(2)$-completeness of the clique problem.

\paragraph{Clique and Independent Set Problems} The \emph{quantum channel $2$-clique problem} with completeness $c$ and soundness $s$ is the promise problem with yes instances consisting of circuits $C$ such that $\Phi_C$ has a $(c,2)$-clique, and no instances consisting of circuits for which $\Phi_C$ has no $(s,2)$-clique. We always assume $c\geq s$, and in the cases we consider that there is some probability gap between them. We define the \emph{quantum channel $2$-independent set problem} in a parallel way. In \cref{def:qcliqueqISProblems}, we express these problems more formally and extend them to larger $k$. These promise problems are written $\ttt{qClique}(k)_{c,s}$ and $\ttt{qIS}(k)_{c,s}$. In this approach, a verifier given a problem instance $C$ may perform computations by running the circuit $C$ in order to implement channel $\Phi_C$, which determines properties of an underlying quantum graph $\mathcal{S}_{\Phi_{C}}$. The size of a problem instance is then the size of the quantum circuit which is used to run this computation, rather than the size of an explicit description for the operator system $\mathcal{S}_{\Phi_{C}}$.

We also consider the restricted problems where the circuits $C$ are only taken from a specific family of quantum circuits $\mc{C}$, which we write as $\ttt{qClique}(k, \mathcal{C})_{c,s}$ and $\ttt{qIS}(k, \mathcal{C})_{c,s}$. This allows us to determine hardness results with respect to particular families of channels/quantum circuits such as entanglement breaking channels and channels which implement partial traces. Importantly, in \cref{Sec:QMA} we specify a family of circuits for which these problems are  $\tsf{QMA}$-complete.

\paragraph{Revisiting the Classical Case.}
In light of the above discussion, we can revisit problems related to cliques and independent sets of classical graphs. More concretely, we consider the input given as a description of a classical circuit which implements a deterministic or probabilistic (noisy) function $f$ on input set $X$. Note that if $X =\lbrace 0,1 \rbrace^n$, then the confusability graph $G_f$ will be an exponentially sized graph. We then define our decision problems as deciding if the corresponding confusability graph $G_f$ has a clique or independent set of size $k$. In this approach, the size of the problem instance is related to the runtime needed to compute the adjacency relations in $G_f$, rather than the size of an explicit description of $G_f$ itself. We also note that our approach also differs from graph problems for succinctly presented graphs~\cite{GW83}, and to the best of our knowledge has not been studied before from this point of view.

In the case that $f$ is deterministic the confusability graph consists of a disjoint union of complete subgraphs. In this case, deciding if $G_f$ has a clique of size $k=2$ is equivalent to the well studied collision problem for functions. It asks whether, given some function $f:X\rightarrow Y$, is there a pair of elements in the domain $x,y\in X$ with the same image $f(x)=f(y)$ --- a collision. The presumed computational difficulty of this problem on certain classes of functions underlies the security of hash functions, which serve as a primary building block of much of modern cryptography.

More formally, we say a deterministic function $f$ has a \emph{$k$-clique}, or a \emph{$k$-independent set}, if there exists distinct inputs $x_1, \dots, x_k$ such that $f(x_1)= \dots =f(x_k)$ or $f(x_1) \neq \dots  \neq f(x_k)$ respectively. 
The corresponding decision problems, whose instances are classical circuits describing functions, are denoted  $\ttt{Clique}(k)$ and $\ttt{IS}(k)$. We study this case in \cref{subsec:DeterministicCase}.

In the case that the underlying function $f$ is probabilistic, \emph{i.e.} $f$ is a noisy classical channel, the corresponding confusability graph $G_f$ can be more interesting. Indeed, by suitably picking the transition probabilities it can be observed that any graph on a vertex set $V$ can arise as the confusability graph of some noisy channel $f$ with input set $V$. In \cref{subsec:Probabilistic}, we study the complexity of deciding the existence of approximate $k$-cliques and approximate $k$-independent sets for probabilistic functions. We restate the relevant definitions for the case $k=2$ here.
\begin{definition} Let $\alpha\in[0,1]$. A probabilistic function $f:X\rightarrow Y$ has an \emph{$(\alpha,2)$-clique} if there exist distinct $x,x' \in X$ such that $\Pr\squ*{f(x)=f(x')}\geq\alpha,$
and an \emph{$(\alpha,2)$-independent set} if there exist distinct $x, x'\in X$ such that $\Pr\squ*{f(x_i)=f(x_j)}\leq1-\alpha.$
\end{definition}

The \emph{probabilistic $2$-clique problem} and the \emph{probabilistic $2$-independent set problem} are the promise problems that consist of deciding the existence of approximate $2$-cliques or $2$-independent sets, up to some specified promise gap $c-s$. In \cref{subsec:Probabilistic} we extend these definitions to larger $k$ and we write $\ttt{pClique}(k)_{c,s}$ and $\ttt{pIS}(k)_{c,s}$ for the corresponding promise problems.

\subsection{Main results}

Our contributions focus on the complexity of the $k$-clique and $k$-independent set problems across four different settings: the quantum case, a restricted quantum case where we only consider entanglement breaking channels, and we conclude by addressing the classical deterministic and probabilistic cases. We show that deciding $2$-cliques across these settings is complete for the classes $\tsf{QMA(2)}$, $\tsf{QMA}$, $\tsf{NP}$, and $\tsf{MA}$, respectively. Further, we show that the $k$-independent set problems in the deterministic, probabilistic, and restricted quantum cases are complete for the same classes.

In the quantum setting we study the complexity of deciding approximate $k$-cliques ($\ttt{qClique}(k)_{c,s}$) and approximate $k$-independent sets ($\ttt{qIS}(k)_{c,s}$)  for a quantum channel $\Phi$. The inputs for these problems are given as the description of a quantum circuit which implements $\Phi$. In \cref{Sec:QMA(2)}, we show that for $k=2$ both problems are in $\tsf{QMA}(2)$, and establish the completeness of $\ttt{qClique}(2)_{c,s}$. Our hardness result is obtained when quantifying over a specific subset of quantum circuits, $\mc{C}_{\Tr}\oplus\mc{C}_{\Tr}\oplus\mc{C}_{\mathrm{EB}}$, as defined in \cref{def:directSums,def:circuit-direct-sum,def:EntanglementBreaking}. Informally, this corresponds to the set of circuits which implement channels that are probabilistic mixtures of entanglement breaking channels and partial trace maps. This is stronger than hardness of the problem quantified over all channels and implies it.

\begin{theorem}[Quantum Case]\label{thrm:MainMain}
       There exist $c,s:\N\rightarrow(0,1)$ with constant gap such that the clique problem $\ttt{qClique}(2,\mc{C}_{\Tr}\oplus\mc{C}_{\Tr}\oplus\mc{C}_{\mathrm{EB}})_{c,s}$ is $\tsf{QMA}(2)$-complete.
\end{theorem}

In \cref{subsec:NewReduction}, we show that, by slightly extending our proof of hardness of the quantum clique problem for $\tsf{QMA}(2)$, we can show that it is also hard for $\tsf{QMA}(k)$. Hence, we provide an alternate proof that $\tsf{QMA}(k)=\tsf{QMA}(2)$ than the original proof of Harrow and Montanaro~\cite{HM13}.

In \cref{Sec:QMA}, we show that by further restricting the set of channels to a particular subset $\mc{C}_{\mathrm{EB}'}$ of the entanglement breaking channels, we can obtain $\tsf{QMA}$-completeness for $\ttt{qClique}(2,\mathcal{C}_{\mathrm{EB}'})_{c,s}$. This suggests any potential distinction between $\tsf{QMA}$ and $\tsf{QMA}(2)$ may be reflected in the difference between the properties of these channels and those considered above.

\begin{theorem}[Restricted Quantum Case]
    There exist $c,s:\N\rightarrow(0,1)$ with polynomial gap such that $\ttt{qClique}(2,\mc{C}_{\mathrm{EB}'})_{c,s}$ is $\tsf{QMA}$-complete.
\end{theorem}

In \cref{Sec:ClassicalCase}, we review the complexity of deciding $k$-cliques and $k$-independent sets for deterministic functions --- the problems $\ttt{Clique}(k)$ and $\ttt{IS}(k)$. These problems can be viewed as a rephrasing of the well-studied collision problem for functions. Our contribution here is to provide new proofs for $\tsf{NP}$-completeness of these problems and to state them in a way which generalizes to the problems considered in \cref{subsec:Probabilistic,Sec:QMA(2),Sec:QMA}.

\begin{theorem}[Deterministic case]
    $\ttt{Clique}(2)$ and $\ttt{IS}(2)$ are $\tsf{NP}$-complete.
\end{theorem}

In \cref{subsec:Probabilistic}, we consider the problems of deciding approximate $k$-cliques ($\ttt{pClique}(k)_{c,s}$) and approximate $k$-independent sets ($\ttt{pIS}(k)_{c,s}$) for probabilistic functions. 

\begin{theorem}[Probabilistic case]
     There exist $c,s:\N\rightarrow(0,1)$ with constant gap such that $\ttt{pClique}(2)_{c,s}$  and $\ttt{pIS}(2)_{c,s}$ are $\tsf{MA}$-complete.
\end{theorem}
Finding natural $\tsf{MA}$-complete problems is in itself a active area of research, and few examples are known \cite{BT10,AG19}.

If one begins with a deterministic or classical noisy channel $f$, and considers its inclusion as a quantum channel $\Phi_f$, as described in \cite{Pau16}, then $f$ will have an $(\alpha,k)$-clique if and only if $\Phi_f$ does. Hence, the classical decision problems analysed in \cref{Sec:ClassicalCase} can naturally be viewed as a restriction of the problems in \cref{Sec:QMA}, which in turn are a restriction of the problems considered and \cref{Sec:QMA(2)}. Our results can thus be viewed as a natural hierarchy of problems which are complete for all four of the proof classes $\tsf{NP}$, $\tsf{MA}$, $\tsf{QMA}$, and $\tsf{QMA}(2)$, and whose complexity may be tuned by varying the collection of channels.

In \cref{sec:a-k-cliques}  we provide a direct reductions from clique and independent set problems in the deterministic, probabilistic, and quantum settings from larger $k$ to $k=2$. In \cref{a-alternateproof}, we provide a more straightforward proof establishing $\tsf{QMA}(2)$-hardness for $\ttt{qClique}(2)_{c,s}$ when quantifying over all quantum circuits than the stronger one given in \cref{thrm:MainMain}.

\subsection{Proof techniques}\label{subsec:ProofTechniques}

Our proof techniques were inspired by techniques in self-testing and protocols for delegating quantum computation, such as in \cite{Gri19,BMZ23}. Other recent work of Bassirian, Fefferman, and Marwaha considers different connections between self-testing and $\tsf{QMA}(2)$~\cite{BFM23}. In short, if we know a channel has or does not have cliques over a particular set of states --- for example, the separable states --- then we can combine it with another channel to guarantee that possible cliques among the remaining states are not important. The proofs use two important constructions, which we will comment on here. We also discuss different techniques we use in the proof of $\tsf{QMA}$-completeness of a restricted version of the problem.

\paragraph{Non-unitary circuits and direct sums} First, we devise a new technique to take direct sums of quantum circuits (\cref{def:circuit-direct-sum}). In order to be able to parametrise arbitrary non-unitary channels by means of circuit diagrams, we work with an extended circuit model that appends partial trace and state preparation maps to a basic set of unitary gates. In this way, we extend the universality of a gate set to CPTP maps by means of the Stinespring dilation theorem. At the level of channels, there are a variety of ways to compose the maps to get new ones, but it is not  always immediate how that translates to composition at the level of circuits --- importantly, whether it is computationally efficient. In our constructions, the way we will combine channels is by means of the direct sum: given two channels $\Phi_1$ and $\Phi_2$ and a probability $p$, the direct sum is $p\Phi_1\oplus(1-p)\Phi_2$, the convex combination of the channels $\Phi_1$ and $\Phi_2$ with orthogonal output spaces. Of course, the circuit model is well-adapted for composition and direct products of maps, but not for sums. To get the wanted superposition of circuits, we might add a qubit in state $\sqrt{p}\ket{0}+\sqrt{1-p}\ket{1}$, conditionally act by either $\Phi_1$ or $\Phi_2$ based on this qubit, and then measure the added qubit. This, however, has a technical obstruction: conditioning a non-unitary channel can put the number of output qubits in superposition or probabilistic mixture, which is untenable. Instead, we use the fact that the circuits decompose into the basic gates to separate the unitary from the non-unitary gates. Then, the unitary gates act conditionally on the additional modulating qubit as wanted and the non-unitary gates can be aligned to act the same way for both $\Phi_1$ and $\Phi_2$. In this way, we get a circuit that implements a direct sum of channels of size polynomial (in fact linear) in the size of the circuits describing the original channels. See also \cref{fig:circuit-sum} for an illustration of this construction. We are optimistic that this direct sum construction might be able to find various other applications in quantum information.

\paragraph{Self-testing for cliques} Next, the direct sum construction allows us to port ideas from self-testing over to questions about quantum cliques. The main technical tool we prove to use this technique is \cref{lem:like-5.5}. As noted above, the technique takes the direct sum of two channels in order to expand the set of pairs of states on which some property of one of the channels holds. To illustrate how this works, we outline the proof of $\tsf{QMA}(2)$-hardness of the quantum $2$-clique problem (\cref{thm:clique-qma2-hard}), which is based on this technique. Given some instance $x$ of a $\tsf{QMA}(2)$ promise problem $(Y,N)$, we know that there is some efficiently constructable circuit $\Phi_x$ that accepts with high probability on some separable state if $x$ is a yes instance, and with low probability on every separable state if $x$ is a no instance. In the first step, we construct an entanglement-breaking channel $\Psi_x$ that runs this verification, measures, and outputs a pure state if it accepts and a highly mixed state if it rejects (\cref{lem:algorithm-channel}). Note that this channel acts on one more qubit than $\Phi_x$, which it immediately traces out, in order to have that qubit guarantee orthogonality of any clique. If $\Phi_x$ accepts with high probability, then $\Psi_x$ outputs a state close to pure, and thus, the overlap of two input states differing only in the orthogonality qubit will be high. As such, if $x$ is a yes instance, $\Psi_x$ has a clique in the separable states. On the other hand, if $\Phi_x$ accepts with low probability, then the output state will be far from pure, resulting in low overlap --- as such, a no instance leads to a channel $\Psi_x$ with no cliques among the separable states. The trick however is to lift this property to all the states. To do so, we use the channel $\Phi_1$ constructed in \cref{lem:close-to-separable,lem:separator}, consisting of a direct sum of partial traces on the subregisters we wish to be unentangled. For this channel, we uncover a property reminiscent of self-testing: if states form a near optimal clique, then they must be separable, since partial traces of entangled states are mixed and hence have low overlap. Finally, in the convex combination of $\Phi_1$ and $\Psi_x$, we find via \cref{lem:like-5.5} that if we weight $\Phi_1$ sufficiently, then the testing for separability overcomes the any potential non-separable cliques of $\Psi_x$, and thus that for a no instance the channel has no cliques. Since the weighting can be chosen to preserve polynomial or constant gaps, this proves the hardness of the $2$-clique problem for $\tsf{QMA}(2)$. We believe that this technique is readily extendable to prove other properties of channels; in fact, in \cref{sec:a-k-cliques}, we use $\Phi_1$ as well as a channel that guarantees orthogonality among the tensor terms of a separable state to reduce the quantum $k$-clique problem directly to the $2$-clique problem.

\paragraph{Independent sets} It is interesting to note that, though this proof technique is very useful in determining properties of cliques, it seems not to be well-adapted to the study of independent sets. First, \cref{lem:like-5.5} is used to expand collections of states where there is no clique, whereas this would require us to be able to expand collections of states where there is no independent set. A more serious technical obstruction, however, is that where cliques may arise from essentially one source, independent sets arise from two. The overlap of two states is high if they are approximately equal \emph{and} pure. However, the overlap of two states is low if they are approximately orthogonal \emph{or} highly mixed. The first corresponds to an approximate version of the independent sets for an operator system, but the second does not. This case of low overlap, which we refer to as a \emph{pseudo-independent} set underpins the proofs of hardness of the clique problem, but seems to be an obstruction for the independent set problem. In particular, if we want to guarantee small overlap, it is not required to output orthogonal states --- it suffices to output a highly mixed state. Also, high overlap is critical in guaranteeing separability of states by means of \cref{lem:separator}. This raises an interesting question: does there exist a channel where the independent sets are separable states? Resolving this would serve to illuminate the complexity of the clique problem greatly.

\paragraph{Entanglement-breaking channels and \textsf{QMA}} Finally, our proof of the $\tsf{QMA}$-completeness of the clique problem for a subclass of the entanglement-breaking channels uses significantly different techniques. The main difficulty lies in showing that this problem actually lies in $\tsf{QMA}$. This relies on an interesting result of Brand\~ao~\cite{Bra08}, which shows that the subclass of $\tsf{QMA}(2)$ where the verifier may measure with a Bell measurement --- measure each of the factors in the separable state separately, and then classically compute whether or not to accept --- lies in $\tsf{QMA}$. We find that computing the clique value by running the channel and then doing the swap test is in fact a Bell measurement when the channel is entanglement-breaking. Nevertheless we must still further restrict, because there is another check needed to verify a clique: checking orthogonality of the states. In our proof that the clique problem is in $\tsf{QMA}(2)$, we also do this via the swap test. However, since this swap test is not prefixed by an entanglement-breaking channel, it cannot be evaluated by a Bell measurement. In fact, Harrow and Montanaro~\cite{HM13} show that the swap test cannot even be run by an LOCC measurement, which is a larger class than the Bell measurements. As such, we need a different method to guarantee orthogonality; we work over channels where the output is independent of one of the qubits, \emph{i.e.} that qubit is immediately traced out. This is the same method as we have used to guarantee existence of an orthogonal clique in the $\tsf{QMA}(2)$-hard clique problem. As such, the prover need not even send as a proof a clique where the states are orthogonal --- the verifier may nevertheless be certain that there exists one. Finally, we note that, unlike the $\tsf{QMA}(2)$-completeness result, the $\tsf{QMA}$-completeness result generalises immediately to the independent set problem as well.

\subsection{Relation to other works}\label{subSec:BG08}
Closely related to our work is that given in the unpublished work of Beigi and Shor \cite{BS08}. In their work, a promise problem referred to as the \emph{quantum clique problem} is proven to be complete for $\tsf{QMA}$. This problem is similar to but distinct from $\ttt{qIS}(k,\mc{C})_{c,s}=(Y,N)$ as we define it in \cref{def:qClique/qIS}. Firstly, in terms of terminology, what Beigi and Shor refer to as a \emph{quantum clique} for a channel is more closely related to what is commonly called an independent set in the quantum graph literature. The standardization of this terminology was established after the release of their work and we chose to model our terminology following more current conventions. 

There are also a few more substantial differences which make our results and theirs incomparable in a direct sense. In order to highlight some of these differences we restate their definition below.

\begin{definition}\label{def:BS08}[Definition 2.9 from \cite{BS08}] Quantum clique problem $(\Phi, k, a, b)$

\begin{itemize}

\item {\bf Input} Integer numbers  $n$ and $k$, non-negative real numbers $a, b$ with an inverse
polynomial gap $b-a > n^{-c}$, and an entanglement breaking
channel $\Phi$ that acts on $n$-qubit states.

\item {\bf Promise} Either there exists $\rho^1\otimes \dots \otimes\rho^k$ such that
$\sum_{i,j} tr(S\,\Phi(\rho^i)\otimes\Phi(\rho^j)) \leq a $ or for
any state $\rho^{1,2\dots k}$ we have $\sum_{i,j}
tr(S\,\Phi^{\otimes 2}(\rho^{i,j})) \geq b$.

\item {\bf Output} Decide which one is the case.

\end{itemize}

\end{definition}

One key difference between \cref{def:qClique/qIS} and \cref{def:BS08} is that their promise problem makes explicit reference to the SWAP operator $S$ --- see \cref{subsec:SWAP} for more details on this and the well-known swap test. Even when quantifying over only entanglement-breaking maps, it is not clear that the two problems are the same, since the swap test can fail if the input state is not given as a product state. That is, the promise made in the work of Beigi and Shor is more significant than we use: no instances are those where the swap test rejects with high probability on all, potentially entangled, states and not just on separable states, as we require.

Another difference is with regard to how the classical description for a quantum channel $\Phi$ is given as input to the problems. Beigi and Shor establish $\tsf{QMA}$-completeness for \cref{def:BS08} when considering entanglement-breaking quantum channels given in terms of an operator sum, or Kraus, representation
\begin{equation} 
\Phi(\rho)=\sum_{i=1}^{r} E_i\rho E_i^\dagger,
\end{equation}
whereas in our case we consider inputs given as a description of a quantum circuit $C$.

Despite some differences in our approaches we took much inspiration from the ideas present in \cite{BS08}. Separate from our work here, another active project is investigating the complexity of quantum graph problems using an approach more closely following that given in \cite{BS08}.

\subsection{Open problems}\label{subSec:OpenProblems}

In this section, we outline a selection of open problems motivated by our work.

\paragraph{Deciding Independent Sets in Quantum Graphs}
The existence of graph complementation provides a straightforward way to reduce the problems of deciding cliques and independent sets for graphs. This also extends to the versions of these problems introduced in \cref{Sec:ClassicalCase}. 

In the quantum setting, our proof technique establishing $\tsf{QMA}(2)$-hardness for the clique problem $\ttt{qClique}(k)_{c,s}$ could not readily be extended to the problem $\ttt{qIS}(k)_{c,s}$. As discussed in \cref{subsec:Qgraphs}, there is no canonical way to take the complement of a quantum graph as there is classically. Instead, several differing approaches have been considered in the literature. Even when one considers the different approaches available, it is not clear how to use these maps on the level of operator systems to obtain an efficient reduction between $\ttt{qClique}(k)_{c,s}$ and $\ttt{qIS}(k)_{c,s}$. Indeed, in an informal sense, it would appear that under any appropriate notion of complementation the complement of a quantum graph can be a significantly more complicated object than the original quantum graph. On the one hand, we know that the quantum independent set problem is contained in $\tsf{QMA}(2)$. On the other hand, it is hard for $\tsf{QMA}$. However, we feel it is unlikely that there exists a $\tsf{QMA}$ verification circuit for $\ttt{qIS}(k)_{c,s}$. This is a similar situation to the pure state version of the $\tsf{QMA}$-complete consistency of local density matrices problem, which has long been conjectured to be complete for $\tsf{QMA}(2)$ \cite{BG22}. It is possible that neither of these problems is $\tsf{QMA}$ or $\tsf{QMA}(2)$-complete, and a reduction from one of the two problems to the other would be enlightening.

\paragraph{Deciding Colourings} Famously, Lov\'{a}sz showed a reduction between the $k$-colouring problem and the $3$-colouring problem in ordinary graphs \cite{Lov73}. In this case, the instances are given using an explicit description of the graph. In \cref{Sec:ClassicalCase}, we introduce new variants of the $k$-clique and $k$-independent set problems motivated from the study of noisy channels and confusability graphs. It is straightforward to similarly define a new version of the $k$-colouring problem. We did not investigate the complexity of this problem but suggest it as worthy area of exploration. Since a colouring on exponentially many vertices cannot be described in polynomial time, it is altogether possible that the class where this colouring problem lives is significantly larger than $\tsf{NP}$.

In the quantum setting, various approaches to define a $k$-colouring for a quantum graph have been considered \cite{KM19,Sta16}. Informally, a $k$-colouring of a quantum graph $\mathcal{S} \subseteq \mathbb{M}_k$ consists of a partition of an orthogonal basis for $\C^k$ into independent sets. As in the classical case, for an $n$-qubit quantum channel, this would require specifying a basis for a $2^n$ dimensional space. It seems unlikely that the corresponding problem could be decided in $\tsf{QMA}(2)$. 

\paragraph{Other Graph Problems:} There are a variety of decision problems for graphs which can possibly be extended to the quantum setting via a similar approach to the one we use for the clique and independent set problems. We suggest two more here. Firstly, in our problems, the size $k$ of a clique or independent set is a parameter which defines the problem. One can instead consider the value $k$ to be specified along with the circuit description as an input. In the classical case of graphs described by their adjacency matrices, this takes the complexity of finding a clique or independent set from $\tsf{P}$ to $\tsf{NP}$-hard \cite{Sip12}. One might imagine a similar jump happens in the cases we have considered, from $\tsf{NP}$, $\tsf{MA}$, $\tsf{QMA}$, or $\tsf{QMA}(2)$ to larger classes.

Secondly, one can consider versions of the well-studied graph isomorphism and non-isomorphism problems. In particular, consider a quantum verifier who receives as input a pair of circuits $C_1$ and $C_2$ which implement classical or quantum channels $\Phi_1$ and $\Phi_2$. The corresponding promise problems would then be to decide if the quantum graphs $\mathcal{S}_{\Phi_{1}}$ and $\mathcal{S}_{\Phi_{1}}$ are approximately isomorphic or far from isomorphic. The quantum analogues of these problems may also provide new approaches to zero-knowledge proof protocols in the quantum setting, since graph isomorphism/non-isomorphism played an important role in early work on zero-knowledge proof systems~\cite{GMW91,BOGG88}. The area of zero-knowledge proof systems in the quantum setting has seen a flurry of recent advances~\cite{GSY19arxiv,VZ19arxiv, BG22, CFGS18, BMZ23}, and it would be interesting if the study of quantum graph isomorphisms can provide new approaches to zero-knowledge in some of these settings.

Another possible area of exploration is Ramsey theory for quantum graphs. Ramsey theory is a branch of mathematics initiated by F. Ramsey which examines the existence of ordered substructures arising in graphs of a sufficient size~\cite{R30}. Ramsey theory has some interesting connections with computational complexity, such as studying the computational complexity of computing the Ramsey numbers~\cite{Bur90}. Several of the motivating themes of Ramsey theory have been explored for quantum graphs~\cite{Wea17,KKS20}, and it may be worthwhile to explore connections to complexity inspired by these works that incorporates our approaches for presenting quantum graphs as circuits.

\subsection{Acknowledgements}

We thank David Cui, Joshua Nevin, Seyed Sajjad Nezhadi, Laura Man\v{c}inska, Se-Jin Kim, and Adrian She for lengthy discussions during the early stages of this project. We also thank Adrian She and David Cui for detailed discussions and for sharing extensive notes on \cite{BS08}.

\subsection{Outline}

In \cref{sec:Prelims} we outline the important notation we use, and give some brief background information on complexity theory and aspects of quantum information we require. In \cref{sec:motivation}, we recall the standard definitions of cliques and independent sets for graphs, and connect them to our definitions of these properties for quantum channels. Next, in \cref{Sec:QMA(2)}, we define and study the clique and independent set problems for quantum channels. We show in \cref{sec:qma2-complete} that the quantum $2$-clique problem is $\tsf{QMA}(2)$-complete, and in \cref{subsec:NewReduction} that this completeness implies a new proof that $\tsf{QMA}(k)=\tsf{QMA}(2)$. In \cref{sec:qis}, we discuss the quantum independent set problem for quantum channels. Finally, in \cref{Sec:QMA}, we present a restriction of the quantum clique and independent set problems that is $\tsf{QMA}$-complete. In \cref{Sec:ClassicalCase}, we consider the classical deterministic (\cref{subsec:DeterministicCase}) and probabilistic (\cref{subsec:Probabilistic}) versions of the clique and independent set problems. We show their completeness for the complexity classes $\tsf{NP}$ and $\tsf{MA}$, respectively. We also collect some additional results in appendices: first, in \cref{sec:a-k-cliques}, we give reductions from $k$-clique and independent set problems to the corresponding problem with $k=2$; also, in \cref{a-alternateproof}, we give a simpler but weaker proof of the $\tsf{QMA}(2)$-completeness of the quantum $2$-clique problem.

\section{Preliminaries}\label{sec:Prelims}

\subsection{Notation}
	
	We use italic capitals $H,K,L$ to represent Hilbert spaces, which we always assume to be finite-dimensional. Given a finite set $S$, write $\C^S$ for the $|S|$-dimensional Hilbert space with canonical (computational) basis $\set*{\ket{s}}{s\in S}$. As a special case, we denote the space of one qubit $Q=\C^{\{0,1\}}$. Fix some Hilbert space $H$. Write $H_\perp$ for the superspace $H+\C\ket{\perp}$ where $\norm{\ket{\perp}}=1$ and $\braket{\perp}{\psi}=0$ for all $\ket{\psi}\in H$. For any subspace $K\subseteq H$, write $\mu_K$ for the maximally-mixed state on $K$. Write $\mc{B}(H)$ for the space of linear operators on $H$, write $\mc{P}(H)$ for the space of positive operators, and $\mc{D}(H)$ to be the set of density operators (trace $1$ positive operators). We use quantum state interchangeably with density operator, specifying to pure quantum states when necessary; and we use quantum channel interchangeably with completely positive trace-preserving (CPTP) map.
	
	We denote the set of all bitstrings $\{0,1\}^\ast:=\bigcup_{n=0}^\infty\{0,1\}^n$. Given some $x\in\{0,1\}^\ast$, write $|x|\in\N$ for its length. For a function $f:\N\rightarrow[0,1]$, write $f:\N\rightarrow(0,1)_{\exp}$ to mean that there exist $N,k>0$ such that $2^{-n^k}<f(n)<1-2^{-n^k}$ for all $n\geq N$. Write the set of polynomially-bounded functions $\ttt{poly}=\set*{f:\N\rightarrow\R}{\exists\,k,N\geq 0.\;f(n)\leq n^k\,\forall\,n\geq N}$. For a natural number $n\in\N$, write $[n]:=\{1,2,\ldots,n\}$.
 
	\subsection{Complexity theory}
	
	A \emph{language} is a subset $L\subseteq\{0,1\}^\ast$. An element $x\in L$ is called a yes instance, and an element $x\notin L$ is called a no instance. A \emph{promise problem} is a pair $(Y,N)$ of disjoint subsets $Y,N\subseteq\{0,1\}^\ast$. Here, $x\in Y$ is called a yes instance and $x\in N$ is called a no instance, while $x\notin Y\cup N$ is indeterminate. Promise problems generalise languages in the sense that every language $L$ corresponds to the promise problem $(L,L^c)$. A \emph{complexity class} is a collection of languages or promise problems.
	
	A function $f:\{0,1\}^\ast\rightarrow\{0,1\}^\ast$ is computable in time $T(n)$ if there exists a Turing machine $M$ that, on input $x$, outputs $M(x)=f(x)$ in $T(|x|)$ steps. We say $f$ is polynomial-time computable if there exists a polynomial $p:\N\rightarrow\N$ such that $f$ is computable in time $p(n)$.
	
	A language $L_1$ is \emph{polynomial-time Karp-reducible} to another language $L_2$ if there exists a polynomial-time computable function $f$ such that $x\in L_1$ iff $f(x)\in L_2$. For promise problems, a similar definition holds, the only difference being that the function must map no instances to no instances as well. Given a complexity class $\tsf{C}$, $L$ is \emph{$\tsf{C}$-hard} (under polynomial-time Karp reductions) if every element of $\tsf{C}$ is polynomial-time Karp reducible to $L$. $L$ is \emph{$\tsf{C}$-complete} if additionally $L\in\tsf{C}$.
	
	\subsection{Quantum and classical circuits}
	
	In its usual usage, the circuit model of quantum computation is used to describe unitary channels. To describe unitaries, we fix a (finite) universal gate set $S\in\mc{U}_{n_0}(\C)$ --- in this work, we do not need to specify the choice of gate set, but one can, for example, always assume that it is the Clifford$+T$ gate set $\{\mathrm{CNOT},H,T\}$. We can approximate any unitary to arbitrary precision by a composition of finitely many tensors of the basic gates.
	
	In order to describe all quantum channels, we need to add two non-unitary operations, that change the number of qubits, to the basic gate set:
	\begin{itemize}
		\item The \emph{partial trace} on one qubit $\Tr:\mc{B}(\C^2)\rightarrow\C$;
		\item \emph{State preparation} of one qubit in the computational basis $\C\rightarrow\mc{B}(\C^2)$, $\lambda\mapsto\lambda\ketbra{0}$.
	\end{itemize}
	Due to the Stinespring dilation theorem, the augmented gate set $S\cup\{\Tr,\ketbra{0}\}$ can approximate any quantum channel to arbitrary precision. Stinespring dilation also guarantees that the circuit can be expressed in a canonical form where qubit preparation happens first, followed by unitary gates, followed by partial traces. In fact, given a circuit description, it can be taken to this canonical form efficiently simply by sliding the state preparations backwards and the partial traces forwards. As such, we will assume that the circuits are always in this form.
	
	Note that POVM measurements can in particular be realised in this model. For example, measurement of a qubit in the computational basis can be realised by the following circuit:
	
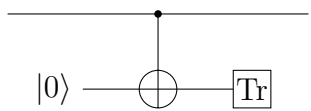
\begin{center}
		\begin{tikzpicture}
			\draw (-1.5,0.5) -- (2.5,0.5) (-0.5,-0.5) node[left]{$\ket{0}$} -- (1.5,-0.5);
			\draw (0.5,-0.5) circle (0.25cm) (0.5,0.5) -- (0.5,-0.75);
			\fill (0.5,0.5) circle (0.05cm);
			\draw (1.5,-0.25) rectangle (2,-0.75) node[pos=0.5]{${\Tr}$};
		\end{tikzpicture}
	\end{center}
	Other measurements can be realised similarly --- diagonalising for PVMs or using Naimark's theorem for general POVMs --- although the circuit representations may consist of many gates. 
	
	For the computational problems we consider, the instances take the form of circuit diagrams describing some quantum channel. It is perhaps simplest from the theory standpoint to present the diagrams as pictures, but it is wholly equivalent to represent the circuits as an encoding of the diagram as a bitstring, which is more natural from a computational standpoint. For a circuit diagram $C$, write $\Phi_C$ for its implementation as a quantum channel. We take the length of the circuit $|C|=\ttt{in}(C)+\ttt{out}(C)+\ttt{gates}(C)$, the sum of the number of input qubits, the number of output qubits, and the number of gates. Note that it is possible to choose an encoding as a bitstring whose length is polynomial in the circuit length described above, and vice versa.
	
	We can also treat classical circuits in a similar way. Here, circuits describe not channels but functions $f:\{0,1\}^m\rightarrow\{0,1\}^n$. But, as in the quantum case, we can choose a finite universal gate set, for example $\{\mathrm{NAND}\}$. Given a classical circuit diagram $C$, write $f_C$ for the function it represents and $|C|=\ttt{in}(C)+\ttt{out}(C)+\ttt{gates}(C)$ for the circuit size.
	
	\subsection{Distances between states}
	
	There are a variety of natural distance measures between quantum states; we use three: the Euclidean norm, the trace norm, and the Frobenius (or Hilbert-Schmidt) norm. The Euclidean distance is only defined on pure states and is simply the natural norm induced by the inner product on the Hilbert space $\norm{\ket{\psi}}=\sqrt{\braket{\psi}}$. The trace norm is defined for all $A\in\mc{B}(H)$ as $\norm{A}_{\Tr}=\frac{1}{2}\Tr\squ{\sqrt{A^\dag A}}$; and the Frobenius norm is defined similarly as $\norm{A}_2=\sqrt{\Tr\squ{A^\dag A}}$. The following simple lemma relates the metrics that these norms induce on the pure states. We use this extensively and without mention.

	\begin{lemma}\label{lem:relate-norms}
		Let $\ket{\psi},\ket{\phi}$ be pure states. Then,
		\begin{align}
			&\norm{\ketbra{\psi}-\ketbra{\phi}}_2=\sqrt{2}\norm{\ketbra{\psi}-\ketbra{\phi}}_{\Tr},\\
			&\norm{\ketbra{\psi}-\ketbra{\phi}}_2\leq\sqrt{2}\norm{\ket{\psi}-\ket{\phi}}.
		\end{align}
	\end{lemma}
	
	\begin{proof}
		For the first relation, let $X=\ketbra{\psi}-\ketbra{\phi}$. If $X=0$, then this holds trivially. Else, there exists an eigenvalue $\lambda\neq 0$ of $X$. Since $\mathrm{rank}(X)=2$ and $\Tr(X)=0$, the only nonzero eigenvalues of $X$ are $\lambda$ and $-\lambda$. This gives that $\norm{X}_2=\sqrt{\lambda^2+(-\lambda)^2}=\sqrt{2}|\lambda|$; and $\norm{X}_{\Tr}=\tfrac{1}{2}\parens*{|\lambda|+|-\lambda|}=|\lambda|$.
		
		For the second relation, note that
		\begin{align}
			\begin{split}
				\norm{\ketbra{\psi}-\ketbra{\phi}}_2^2&=2-2\abs*{\braket{\psi}{\phi}}^2\leq2(2-2\abs{\braket{\psi}{\phi}})\\
				&\leq2(2-2\latRe\braket{\psi}{\phi})=2\norm{\ket{\psi}-\ket{\phi}}^2.
			\end{split}
		\end{align}
	\end{proof}

	\subsection{The swap test}\label{subsec:SWAP}
	
	The swap test, first introduced in~\cite{BCWdW01}, provides a computationally efficient way to project on the symmetric subspace of a tensor product Hilbert space --- the $+1$ eigenspace of the swap operator that exchanges the two tensor terms. The swap test proceeds on a state $\rho\in\mc{D}(H\otimes H)$ as follows:
	\begin{enumerate}[1.]
		\item Prepare an auxiliary qubit in the state $\ket{+}$.
		\item Act with the swap operator on $\rho$ conditionally on the auxiliary qubit.
		\item Measure the auxiliary qubit in the Hadamard basis. Accept if the measurement result is $0$ and reject if not.
	\end{enumerate}
	The swap test is illustrated as a quantum circuit \cref{fig:swap-test}.
	
	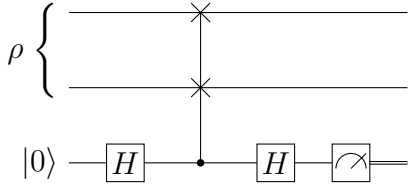
\begin{figure}
		\centering
		\begin{tikzpicture}
			\draw (0,1) -- (4.5,1) (0,0) -- (4.5,0) (0,-1) -- (0.5,-1) (1,-1) -- (2.5,-1) (3,-1) -- (3.5,-1);
			\draw[double] (4,-1) -- (4.5,-1);
			\draw (1.75,1) -- (1.75,-1) (1.625,0.875) -- (1.875,1.125) (1.625,1.125) -- (1.875,0.875) (1.625,-0.125) -- (1.875,0.125) (1.625,0.125) -- (1.875,-0.125);
			\draw (0.5,-1.25) rectangle (1,-0.75) node[pos=0.5]{$H$};
			\draw (2.5,-1.25) rectangle (3,-0.75) node[pos=0.5]{$H$};
			\draw (3.5,-1.25) rectangle (4,-0.75) (3.75,-0.875) arc (90:0:0.2) (3.75,-0.875) arc (90:180:0.2) (3.75,-1.075) -- (3.95,-0.875);
			\fill (1.75,-1) circle (0.05cm);
			\node[left] at (0,0.5){$\rho\;{\Bigg\{}$};
			\node[left] at (0,-1){$\ket{0}$};
		\end{tikzpicture}
		
		\caption{Quantum circuit representation of the swap test.}
		\label{fig:swap-test}
	\end{figure}
	
	Let $\Pi_H:H\otimes H\rightarrow H\otimes H$ be the projector onto the symmetric subspace. Then, the swap test implements the channel $\rho\mapsto\Pi_H\rho\Pi_H\otimes\ketbra{0}+(\Id-\Pi_H)\rho(\Id-\Pi_H)\otimes\ketbra{1}$. It's also useful to note that, if $\rho$ is a separable state $\rho=\rho_1\otimes\rho_2$, then the probability of the swap test accepting is $\frac{1}{2}+\frac{1}{2}\Tr\squ*{\rho_1\rho_2}$.

\section{Graph Parameters for Quantum Channels}\label{sec:motivation}

In this section, we recall the definitions of quantum graphs and their graph parameters. Since every quantum graph can be generated by a quantum channel, we discuss how the clique and independent set graph parameters can be pulled back to the level of quantum channels. This serves to motivate the approximate version of these parameters that we introduce in~\cref{def:qClique/qIS}.

\begin{definition}
    A \emph{quantum graph} is an operator system $S\subseteq\mc{B}(H)$ for $H$ a finite-dimensional Hilbert space, that is a vector subspace closed under the adjoint and containing $I$. 
\end{definition}

As noted in the introduction, any quantum channel generates a quantum graph, sometimes referred to as its quantum confusability graph. For $\Phi:\mc{B}(H)\rightarrow\mc{B}(K)$ with Kraus decomposition $\Phi(\rho)=\sum_{i=1}^n A_i\rho A_i^\ast$, the corresponding quantum graph is
$$S_\Phi=\spn_\C\set*{A_i^\ast A_j}{1\leq i,j\leq n}.$$

\begin{proposition}[{\cite[Proposition 7.18]{Pau16}}]
    Let $S\subseteq\mc{B}(H)$ be a quantum graph. Then, there exists a finite-dimensional Hilbert space $K$ and a quantum channel $\Phi:\mc{B}(H)\rightarrow\mc{B}(K)$ such that~$S=S_\Phi$.
\end{proposition}

However, as for classical confusability  graphs, the quantum channel corresponding to a quantum graph is not unique. Next, we recall the definitions of the independent set and clique graph parameters for a quantum graph.

\begin{definition}[\cite{DSW13,BTW21}]
    Let $S\subseteq\mc{B}(H)$ be an operator system. A collection of orthogonal states $\ket{v_1},\ldots,\ket{v_k}\in H$ form a
    \begin{itemize}
        \item $k$-\emph{independent set} of $S$ if $\ketbra{v_i}{v_j}$ is orthogonal to $S$ for all $i\neq j$;
        \item $k$-\emph{clique} of $S$ if $\ketbra{v_i}{v_j}\in S$ for all $i\neq j$.
    \end{itemize}
\end{definition}

As for the classical case, these graph parameters can also be described in terms of a channel generating $S$. This will allow us to work with quantum graph parameters for quantum channels.

In the case of independent sets, Duan, Severini, and Winter establish that the independent sets for $S$ correspond to states which are perfectly distinguishable after passing through the channel.

\begin{lemma}[\cite{DSW13}]
    Let $\Phi:\mc{B}(H)\rightarrow\mc{B}(K)$ be a quantum channel and $S=S_\Phi\subseteq\mc{B}(H)$ be its quantum confusability graph. Then, pure states $\ket{v_1},\ldots,\ket{v_k}\in H$ are a $k$-independent set of $S$ if and only if the overlap
    \begin{align*}
        \Tr\squ*{\Phi(\ketbra{v_i})\Phi(\ketbra{v_j})}=0
    \end{align*}
    for all $i\neq j$.
\end{lemma}

The connection between cliques for quantum graphs and channels is not as direct, since in general the quantum graph corresponding to a channel does not retain enough information about the original channel. Nevertheless given a quantum graph $S$ one can associate to it a channel $\Phi$ whose behavior determines cliques in the quantum graph $S$. 

\begin{lemma}\label{lem:Construction1}
    Let $S\subseteq\mc{B}(H)$ be an operator system. Then, there exists a Hilbert space $K$, a channel $\Phi:\mc{B}(H)\rightarrow\mc{B}(K)$, and a constant $c>0$ (depending on the operator system) such that orthogonal states $\ket{u},\ket{v}\in H$ satisfy $\ketbra{u}{v}\in S$ if and only if
    $$\Tr\squ*{\Phi(\ketbra{u})\Phi(\ketbra{v})}\geq c.$$
\end{lemma}

\begin{proof}
    First we construct the requisite channel, following the construction of Paulsen~\cite{Pau16}. Write $d=\dim H$ and $n=\dim S$. Let $X_1,\ldots,X_{n-1}\in\mc{B}(H)$ be hermitian operators such that $\{I/\sqrt{d},X_1,\ldots,X_n\}$ is an orthonormal basis of $S$. Then, set $H=\sum_{i=1}^{n-1}X_i\otimes\parens*{\ketbra{i}{i+1}+\ketbra{i+1}{i}}$. $H$ is hermitian, so there exists $r>0$ such that $rI+H\geq 0$. Let $C\in\mbb{M}_n(\mc{B}(H))$ be such that $C^\ast C=\frac{1}{nr}(rI+H)$ and let $\{C_i\}_{i=1}^n$ be the $dn\times d$ block columns of $C$. Define $\Phi:\mc{B}(H)\rightarrow\mbb{M}_n(\mc{B}(H))$ as $\Phi(\rho)=\sum_iC_i\rho C_i^\ast$. Note that $\Phi$  is a quantum channel as $\sum_i C_i^\ast C_i=\sum_i\frac{I}{n}=I$. Also, $S=S_\Phi$ since $$C_i^\ast C_j=\begin{cases}I/n&i=j\\X_i/(nr)&j=i+1\\X_{i+1}/(nr)&j=i-1\\0&\text{ else}\end{cases},$$
    giving $S_\Phi=\spn\{C_i^\ast C_j\}=S$. Also, take $c=\frac{2}{(nr)^2}$.

    Let $\ket{u},\ket{v}\in H$ be orthogonal pure states. Then, the overlap
    \begin{align*}
        \Tr\squ*{\Phi(\ketbra{u})\Phi(\ketbra{v})}&=\sum_{i,j}\Tr\squ*{C_i\ketbra{u}C_i^\ast C_j\ketbra{v}C_j^\ast}=\sum_{i,j}\abs*{\Tr\parens*{C_i^\ast C_j\ketbra{v}{u}}}^2\\
        &=\frac{2}{(nr)^2}\sum_{i}\abs*{\Tr\parens*{X_i\ketbra{v}{u}}}^2=c\norm*{\mathrm{proj}_S\ketbra{u}{v}}^2_2,
    \end{align*}
    the length of the orthogonal projection onto $S$ with respect to the Hilbert-Schmidt inner product. Then, if $\ketbra{u}{v}\in S$, $\norm*{\mathrm{proj}_S\ketbra{u}{v}}^2_2=1$ and if $\ketbra{u}{v}\notin S$, $\norm*{\mathrm{proj}_S\ketbra{u}{v}}^2_2<1$, giving that $\Tr\squ*{\Phi(\ketbra{u})\Phi(\ketbra{v})}\geq c$ if and only if $\ketbra{u}{v}\in S$.
\end{proof}

We use~\cref{lem:Construction1} to motivate the notion of a clique (\emph{cf.} independent set) for a quantum channel $\Phi$, as being a collection of orthogonal states with sufficiently high average (\emph{cf.} low) overlap after passing through the channel. 

\begin{definition}\label{def:qClique/qIS} Let $H, K$ be Hilbert spaces and let $k\in\N$. 
    \begin{itemize}
        \item A quantum channel $\Phi:\mc{B}(H)\rightarrow\mc{B}(K)$ has an \emph{$(\alpha,k)$-clique} if there exist orthogonal states $\rho_1,\ldots,\rho_k\in\mc{D}(H)$ such that
        \begin{align}
            \frac{2}{k(k-1)}\sum_{1\leq i<j\leq k}\Tr\squ*{\Phi(\rho_i)\Phi(\rho_j)}\geq\alpha.
        \end{align}

        \item A quantum channel $\Phi:\mc{B}(H)\rightarrow\mc{B}(K)$ has an \emph{$(\alpha,k)$-independent set} if there exist orthogonal states $\rho_1,\ldots,\rho_k\in\mc{D}(H)$ such that
        \begin{align}
            \frac{2}{k(k-1)}\sum_{1\leq i<j\leq k}\Tr\squ*{\Phi(\rho_i)\Phi(\rho_j)}\leq1-\alpha.
        \end{align}
    \end{itemize}
\end{definition}

\section{The Complexity of the Clique Problem}\label{Sec:QMA(2)}

\subsection{The quantum clique problem is \textsf{QMA}(2)-complete}\label{sec:qma2-complete}

In this section, we introduce and study the complexity of the promise problem of deciding the existence of an approximate independent set/clique for a quantum channel, given as a circuit.

\begin{definition}\label{def:qcliqueqISProblems}
    Let $\mc{C}$ be some collection of quantum circuits, $c,s:\N\rightarrow[0,1]$, and $k\in\N$. The \emph{quantum channel $k$-clique problem} on $\mc{C}$ with completeness $c$ and soundness $s$ is the promise problem $\ttt{qClique}(k,\mc{C})_{c,s}=(Y,N)$ with
    \begin{align}
    \begin{split}
        &Y=\set*{\text{quantum circuits }C}{C\in\mc{C},\,\Phi_C\text{ has a $(c(|C|),k)$-clique}},\\
        &N=\set*{\text{quantum circuits }C}{C\in\mc{C},\,\Phi_C\text{ has no $(s(|C|),k)$-clique}},
    \end{split}
    \end{align}
    
    The \emph{quantum channel $k$-independent set problem} on $\mc{C}$ with completeness $c$ and soundness $s$ is the promise problem $\ttt{qIS}(k,\mc{C})_{c,s}=(Y,N)$ with
    \begin{align}
    \begin{split}
        &Y=\set*{\text{quantum circuits }C}{\Phi_C\in\mc{C},\,\Phi_C\text{ has a $(c(|C|),k)$-independent set}},\\
        &N=\set*{\text{quantum circuits }C}{\Phi_C\in\mc{C},\,\Phi_C\text{ has no $(s(|C|),k)$-independent set}},
    \end{split}
    \end{align}

    Write $\ttt{qClique}(k)_{c,s}$ and $\ttt{qIS}(k)_{c,s}$ for the promise problems with respect to all circuits.
\end{definition}

The main results of this section are to show that all of the above problems are in $\tsf{QMA}(2)$ for $c,s:\N\rightarrow(0,1)_{\exp}$ with polynomial gap, and to show that the quantum $2$-clique problem for some class of channels is in fact $\tsf{QMA}(2)$-hard. First, we recall the formal definition of $\tsf{QMA}(2)$.

\begin{definition}
    Let $c,s:\N\rightarrow[0,1]$. A promise problem $Y,N\subseteq\{0,1\}^\ast$ is in $\tsf{QMA}(2)_{c,s}$ if there exist polynomials $p,q:\N\rightarrow\N$ and a Turing machine $V$ with one input tape and one output tape such that
    \begin{itemize}
        \item For all $x\in\{0,1\}^\ast$, $V$ halts on input $x$ in $q(|x|)$ steps and outputs the description of a quantum circuit from $2p(|x|)$ qubits to one qubit.

        \item For all $x\in Y$, there exist states $\rho,\sigma\in\mc{D}(Q^{\otimes p(|x|)})$ such that $\braket{1}{\Phi_{V(x)}(\rho\otimes\sigma)}{1}\geq c(|x|)$.

        \item For all $x\in N$ and $\rho,\sigma\in\mc{D}(Q^{\otimes p(|x|)})$, $\braket{1}{\Phi_{V(x)}(\rho\otimes\sigma)}{1}\leq s(|x|)$.
    \end{itemize}
\end{definition}

We know that whenever $c,s:\N\rightarrow(0,1)_{\exp}$ and have polynomial gap, then $\tsf{QMA}(2)_{c,s}=\tsf{QMA}(2)_{\frac{2}{3},\frac{1}{3}}=:\tsf{QMA}(2)$, as for $\tsf{MA}$ and $\tsf{QMA}$. Due to \cite{HM13}, we also know that $\tsf{QMA}(P)=\tsf{QMA}(2)$ for $P:\N\rightarrow\N$ any polynomial number of provers.

\begin{lemma}\label{lem:in-qma2}
    Let $\mc{C}$ be any set of quantum circuits and let $c,s:\N\rightarrow[0,1]$. Then, the problem $\ttt{qClique}(2,\mc{C})_{c,s}\in\tsf{QMA}(2)_{c',s'}$, where $c'=\frac{1}{2}+(c-s)\frac{c}{8}$ and $s'=\frac{1}{2}+(c-s)\frac{c+s}{16}$
\end{lemma}
Thus, if $c,s:\N\rightarrow(0,1)_{\exp}$ have polynomial gap, then $c'-s'=\frac{(c-s)^2}{16}\in\frac{1}{\mathrm{poly}}$, giving that $\ttt{qClique}(2,\mc{C})_{c,s}\in\tsf{QMA}(2)$.

\begin{proof}
    Suppose the verifier receives a quantum circuit $C$. Let $\Phi=\Phi_C$. The verifier expects to receive a proof $\rho\otimes\sigma\in \mc{D}(Q^{\otimes 2\ttt{in}(C)})$, and then effects the following procedure.
    \begin{enumerate}[1.]
        \item The verifier samples a random bit $b$, which is $0$ with probability $p$ (to be specified later).

        \item If $b=0$, the verifier runs the swap test on $\rho\otimes\sigma$. If it fails, she accepts and outputs $1$, and if it passes she rejects and outputs $0$.

        \item If $b=1$, the verifier acts with $\Phi$ to get $\Phi(\rho)\otimes\Phi(\sigma)$. Then, she runs the swap test on this state, and outputs $1$ if it passes and $0$ if it fails.
    \end{enumerate}
    Now, suppose $C\in Y$. Then, there exists a $(c(|C|),2)$-clique $\rho,\sigma\in\mc{D}(Q^{\ttt{in}(C)})$, so the prover provides $\rho\otimes\sigma$ to the verifier. If $b=0$, the probability of accepting is $\frac{1}{2}-\frac{1}{2}\Tr(\rho\sigma)=\frac{1}{2}$, as $\rho$ and $\sigma$ are orthogonal; if $b=1$, since $\rho$ and $\sigma$ form a clique, the probability of accepting is $\frac{1}{2}+\frac{1}{2}\Tr\squ*{\Phi(\rho)\Phi(\sigma)}\geq\frac{1}{2}+\frac{c(|C|)}{2}$. All together, the acceptance probability is at least $c'(|C|)=p\frac{1}{2}+(1-p)\parens*{\frac{1}{2}+\frac{c(|C|)}{2}}=\frac{1}{2}+(1-p)\frac{c(|C|)}{4}$. On the other hand, suppose that $C\in N$. Then, for all orthogonal $\rho,\sigma\in\mc{D}(Q^{\ttt{in}(C)})$, $\Tr\squ*{\Phi(\rho)\Phi(\sigma)}\leq s(|C|)$. Without loss of generality, we may assume that the verifier receives a pure state proof $\rho\otimes\sigma=\ketbra{\psi}\otimes\ketbra{\phi}$. Then, if $b=0$, the probability of accepting is $\frac{1}{2}-\frac{1}{2}|\braket{\psi}{\phi}|^2$. There exists a state $\ket{\phi'}$ orthogonal to $\ket{\psi}$ such that $\ket{\phi}=\alpha\ket{\psi}+\beta\ket{\phi'}$ with $|\alpha|^2=|\braket{\psi}{\phi}|^2$. Now, note that
    \begin{align}
        \norm{\ketbra{\phi}-\ketbra{\phi'}}_{\Tr}=\frac{1}{\sqrt{2}}\norm{\ketbra{\phi}-\ketbra{\phi'}}_2=\sqrt{\frac{1}{2}\parens*{2-2|\beta|^2}}=|\braket{\phi}{\psi}|.
    \end{align}
    So, if $b=1$, the probability of accepting is
    \begin{align}
    \begin{split}
        \frac{1}{2}+\frac{1}{2}\Tr\squ*{\Phi(\ketbra{\psi})\Phi(\ketbra{\phi})}&\leq\frac{1}{2}+\frac{1}{2}\Tr\squ*{\Phi(\ketbra{\psi})\Phi(\ketbra{\phi'})}+\frac{1}{2}\norm{\ketbra{\phi}-\ketbra{\phi'}}_{\Tr}\\
        &\leq\frac{1}{2}+\frac{s(|C|)}{2}+\frac{|\braket{\psi}{\phi}|}{2}.
    \end{split}
    \end{align}
    Thus, the total probability of accepting is at most
    \begin{align}
    \begin{split}
        s'(|C|)&=\frac{1}{2}+(1-p)\frac{s(|C|)}{4}+\frac{1}{2}\parens*{(1-p)|\braket{\psi}{\phi}|-p|\braket{\psi}{\phi}|^2}\\
        &\leq\frac{1}{2}+(1-p)\frac{s(|C|)}{4}+\frac{(1-p)^2}{8p},
    \end{split}
    \end{align}
    by maximising over $|\braket{\phi}{\psi}|\in\R$. Now taking $p=1-\frac{c(|C|)-s(|C|)}{2}$, we get $c'=\frac{1}{2}+(c-s)\frac{c}{8}$ and $s'=\frac{1}{2}+(c-s)\frac{s}{8}+\frac{(c-s)^2}{16(2-(c-s))}\leq\frac{1}{2}+(c-s)\frac{s}{8}+\frac{(c-s)^2}{16}=\frac{1}{2}+(c-s)\frac{c+s}{16}$
\end{proof}

Next, we want to show that the quantum clique problem is in fact $\tsf{QMA}(2)$-hard as well. In fact, we show the stronger property that a subclass of all channels achieves this hardness. First, we introduce this subclass.

\begin{definition}\label{def:EntanglementBreaking}
    A quantum channel $\Phi:\mc{B}(H)\rightarrow\mc{B}(K)$ is \emph{entanglement-breaking} if there exists a POVM $\{M_i\}_i\subseteq\mc{P}(H)$ and a set of states $\{\sigma_i\}_i\subseteq\mc{P}(K)$ such that $\Phi(\rho)=\sum_i\Tr\squ*{M_i\rho}\sigma_i$.
    
    Write $\mc{C}_{\mathrm{EB}}$ for the set of quantum circuits representing entanglement-breaking channels. Write $\mc{C}_{\Tr}$ for the collection of circuits that are partial traces on some of the qubits.
\end{definition}

The name \emph{entanglement-breaking} for these maps arises since, if $\Phi$ is entanglement-breaking, then the output state of $\Id\otimes\Phi$ is always separable. These are the only maps with the property. For more on entanglement-breaking channels, see for example~\cite{HSR03,Wat18}.

The circuits we are interested in combine the above two restricted classes of channels by means of a direct sum. At the level of channels, we mean to use channels of the form $p\Phi_{\Tr}\oplus(1-p)\Phi_{\mathrm{EB}}$. To do so, we need a way to form the direct sum of two quantum circuits.

\begin{definition}\label{def:circuit-direct-sum}
	Let $C_1$ and $C_2$ be circuits with $\ttt{in}(C_1)=\ttt{in}(C_2)$ and $p\in[0,1]$. Let $P_i$ be the number of qubits prepared by $C_i$, and let $T_i=\ttt{in}(C_i)+P_i-\ttt{out}(C_i)$ be the number of qubits traced out by $C_i$. We define the \emph{direct sum circuit} $C_1\oplus_p C_2$ as the circuit constructed as follows
	\begin{enumerate}[(i)]
		\item The circuit prepares a qubit $R$ in the state $\sqrt{p}\ket{0}+\sqrt{1-p}\ket{1}$.
		
		\item The circuit prepares $\max\{P_1,P_2\}+|\ttt{out}(C_1)-\ttt{out}(C_2)|$ qubits in state $\ket{0}$.
		
		\item The circuit acts with all the unitary gates of $C_1$ conditioned on $R$. Then, it swaps, conditionally on $R$ as well, the qubits of $C_1$ that would be traced out to the last positions.
		
		\item The circuit flips $R$ with $X$ and acts in the same way as above for $C_2$ conditioned on $R$, before flipping back with $X$.
		
		\item The circuit traces out the last $\max\{T_1,T_2\}$ qubits.
		
		\item The circuit measures $R$ in the computational basis.
	\end{enumerate}
\end{definition}

Note that the state $\sqrt{p}\ket{0}+\sqrt{1-p}\ket{1}$ cannot be generated exactly by finitely many gates from the gate set for all but countably many values of $p$. However, for any $p$, polynomially-many gates (in the circuit sizes) can be used to construct an exponentially-good approximation to the state, which we use to keep the circuits of reasonable length. See also \cref{fig:circuit-sum} for an example of the implementation of the above definition.
	
	\begin{figure}
		\centering
		\begin{tikzpicture}
			\node at (-4,-0.5){$C_1=$};
			\node at (-1.25,0.25){$\vdots$};
			\node at (-0.25,-1){$\vdots$};
			\node at (3.25,0.25){$\vdots$};
			\node at (2.25,-1){$\vdots$};
			\node at (1,-0.5){$\cdots$};
			\node at (-2.5,0.125){$\ttt{in}(C_1)\,{\Big\{}$};
			\node at (-2,-1.125){$P_1\,{\Big\{}$};
			\node at (4.5,0.125){${\Big\}}\ttt{out}(C_1)$};
			\node at (4,-1.125){${\Big\}}T_1$};
			\draw (-1.5,0.5) -- (0.625,0.5) (1.375,0.5) -- (1.5,0.5) (2,0.5) -- (3.5,0.5) (-1.5,-0.25) -- (0,-0.25) (0.5,-0.25) -- (0.625,-0.25) (1.375,-0.25) -- (3.5,-0.25) (-0.5,-0.75) node[left]{$\ket{0}$} -- (0,-0.75) (0.5,-0.75) -- (0.625,-0.75) (1.375,-0.75) -- (2.5,-0.75) (-0.5,-1.5) node[left]{$\ket{0}$} -- (0.625,-1.5) (1.375,-1.5) -- (2.5,-1.5);
			\draw (0,-0.125) rectangle (0.5,-0.875) node[pos=0.5]{$U_1$};
			\draw (1.5,0.25) rectangle (2,0.75) node[pos=0.5]{$U_k$};
			\draw (2.5,-0.5) rectangle (3,-1) node[pos=0.5]{${\Tr}$};
			\draw (2.5,-1.25) rectangle (3,-1.75) node[pos=0.5]{${\Tr}$};
		\end{tikzpicture}
		
		\vspace{0.5cm}
	
		\begin{tikzpicture}
			\node at (-4,-0.5){$C_2=$};
			\node at (-1.25,0.25){$\vdots$};
			\node at (-0.25,-1){$\vdots$};
			\node at (3.25,0.25){$\vdots$};
			\node at (2.25,-1){$\vdots$};
			\node at (1,-0.5){$\cdots$};
			\node at (-2.5,0.125){$\ttt{in}(C_2)\,{\Big\{}$};
			\node at (-2,-1.125){$P_2\,{\Big\{}$};
			\node at (4.5,0.125){${\Big\}}\ttt{out}(C_2)$};
			\node at (4,-1.125){${\Big\}}T_2$};
			\draw (-1.5,0.5) -- (0,0.5) (0.5,0.5) -- (0.625,0.5) (1.375,0.5) -- (3.5,0.5) (-1.5,-0.25) -- (0.625,-0.25) (1.375,-0.25) -- (1.5,-0.25) (2,-0.25) -- (3.5,-0.25) (-0.5,-0.75) node[left]{$\ket{0}$} -- (0.625,-0.75) (1.375,-0.75) -- (2.5,-0.75) (-0.5,-1.5) node[left]{$\ket{0}$} -- (0.625,-1.5) (1.375,-1.5) -- (2.5,-1.5);
			\draw (0,0.675) rectangle (0.5,-0.125) node[pos=0.5]{$V_1$};
			\draw (1.5,-0.5) rectangle (2,0) node[pos=0.5]{$V_l$};
			\draw (2.5,-0.5) rectangle (3,-1) node[pos=0.5]{${\Tr}$};
			\draw (2.5,-1.25) rectangle (3,-1.75) node[pos=0.5]{${\Tr}$};
		\end{tikzpicture}
		
		\vspace{0.5cm}
	
		\begin{tikzpicture}
			\node at (-3,0){$C_1\oplus_p C_2=$};
			\node at (-1.25,0.25){$\vdots$};
			\node at (-0.25,-1){$\vdots$};
			\node at (7.75,0.25){$\vdots$};
			\node at (2.75,-1){$\vdots$};
			\node at (6.75,-1){$\vdots$};
			\node at (1.5,-0.5){$\cdots$};
			\node at (4,-0.5){$\cdots$};
			\draw (-0.5,1) node[left]{$\ket{0}$} -- (1.125,1) (1.875,1) -- (3.625,1) (4.375,1) -- (7,1) (-0.5,1.5) node[left]{$\ket{0}$} -- (0,1.5) (0.5,1.5) -- (1.125,1.5) (1.875,1.5) -- (2.5,1.5) (3,1.5) -- (3.625,1.5) (4.375,1.5) -- (5.5,1.5) (6,1.5) -- (8,1.5);
			\draw (-1.5,0.5) -- (1.125,0.5) (1.875,0.5) -- (2,0.5) (2.5,0.5) -- (3,0.5) (3.5,0.5) -- (3.625,0.5) (4.375,0.5) -- (8,0.5) (-1.5,-0.25) -- (0.5,-0.25) (1,-0.25) -- (1.125,-0.25) (1.875,-0.25) -- (3.625,-0.25) (4.375,-0.25) -- (8,-0.25) (-0.5,-0.75) node[left]{$\ket{0}$} -- (0.5,-0.75) (1,-0.75) -- (1.125,-0.75) (1.875,-0.75) -- (3.625,-0.75) (4.375,-0.75) -- (7,-0.75) (-0.5,-1.5) node[left]{$\ket{0}$} -- (1.125,-1.5) (1.875,-1.5) -- (3.625,-1.5) (4.375,-1.5) -- (7,-1.5);
			\draw (0,1.25) rectangle (0.5,1.75) node[pos=0.5]{$U_p$};
			\draw (0.5,-0.125) rectangle (1,-0.875) node[pos=0.5]{$U_1$} (0.75,-0.125) -- (0.75,1.5);
			\fill (0.75,1.5) circle (0.05cm);
			\draw (2,0.25) rectangle (2.5,0.75) node[pos=0.5]{$U_k$} (2.25,0.75) -- (2.25,1.5);
			\fill (2.25,1.5) circle (0.05cm);
			\draw (2.5,1.25) rectangle (3,1.75) node[pos=0.5]{$X$};
			\draw (3,0.675) rectangle (3.5,-0.125) node[pos=0.5]{$V_1$} (3.25,0.675) -- (3.25,1.5);
			\fill (3.25,1.5) circle (0.05cm);
			\draw (4.5,-0.2) rectangle (5,0.3) node[pos=0.5]{$V_l$} (4.75,0.3) -- (4.75,1.5);
			\fill (4.75,1.5) circle (0.05cm);
			\draw (5.25,-1.5) -- (5.25,1.5) (5.125,-0.125) -- (5.375,-0.375) (5.125,-0.375) -- (5.375,-0.125) (5.125,-1.375) -- (5.375,-1.625) (5.125,-1.625) -- (5.375,-1.375);
			\fill (5.25,1.5) circle (0.05cm);
			\draw (5.5,1.25) rectangle (6,1.75) node[pos=0.5]{$X$};
			\draw (6.5,1) circle (0.25cm) (6.5,0.75) -- (6.5,1.5);
			\fill (6.5,1.5) circle (0.05cm);
			\draw (7,0.75) rectangle (7.5,1.25) node[pos=0.5]{${\Tr}$};
			\draw (7,-0.5) rectangle (7.5,-1) node[pos=0.5]{${\Tr}$};
			\draw (7,-1.25) rectangle (7.5,-1.75) node[pos=0.5]{${\Tr}$};
		\end{tikzpicture}
		
		\caption{Representation of the construction of the circuit $C_1\oplus_p C_2$ from the circuits $C_1$ and $C_2$ in canonical form, with $\ttt{out}(C_2)<\ttt{out}(C_1)$. $U_p$ is a unitary that implements the map $\ket{0}\mapsto\sqrt{p}\ket{0}+\sqrt{1-p}\ket{1}$ to good approximation.}
		\label{fig:circuit-sum}
	\end{figure}
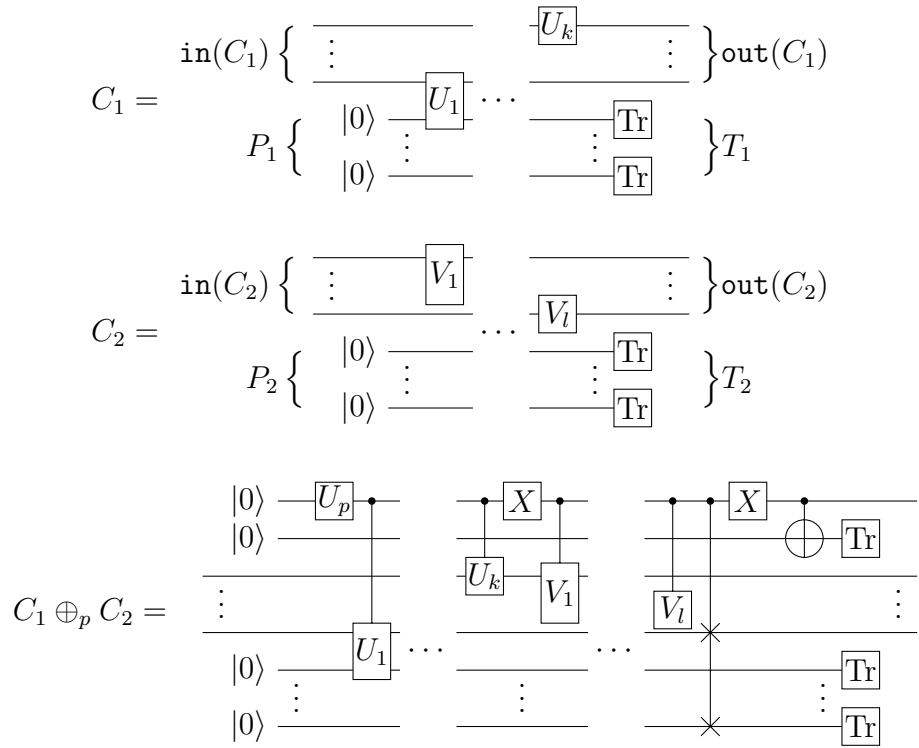

Using the natural identification, for $m>n$, of $Q^{\otimes m}\oplus Q^{\otimes n}$ with the subspace $\ket{0}\otimes Q^{\otimes m}+\ket{1}\otimes Q^{\otimes n}\otimes\ket{0^{m-n}}\subseteq Q^{\otimes(m+1)}$, we see that $\Phi_{C_1\oplus_p C_2}=p\Phi_{C_1}\oplus(1-p)\Phi_{C_2}$. Note also that, by the above definition $|C_1\oplus_p C_2|=O(|C_1|+|C_2|)$.

\begin{definition}\label{def:directSums}
    Let $\mc{C}_1$ and $\mc{C}_2$ be collections of quantum circuits. Write
    \begin{align}
        \mc{C}_1\oplus\mc{C}_2=\set*{C_1\oplus_p C_2}{C_1\in\mc{C}_1,\,C_2\in\mc{C}_2,\,p\in[0,1]}.
    \end{align}
\end{definition}

\begin{theorem}\label{thm:clique-qma2-hard}
    There exist $c,s:\N\rightarrow(0,1)_{\exp}$ with constant gap such that the clique problem $\ttt{qClique}(2)_{c,s}(2,\mc{C}_{\Tr}\oplus\mc{C}_{\Tr}\oplus\mc{C}_{\mathrm{EB}})$ is $\tsf{QMA}(2)$-hard.
\end{theorem}

To prove this theorem, we develop a technique to take a channel whose cliques have certain properties with respect to a set of states, and combine it with another channel in order to enlarge the set of states on which those properties hold. This is inspired by self-testing techniques, such as Lemma 5.5 of \cite{BMZ23}.

\begin{lemma}\label{lem:like-5.5}
    Let $H$, $K$, and $K'$ be Hilbert spaces, and $A,B\subseteq\mc{D}(H)^2$ be sets of pairs of states. Suppose there exist channels $\Phi:\mc{B}(H)\rightarrow\mc{B}(K)$ and $\Psi:\mc{B}(H)\rightarrow\mc{B}(K')$ with the properties that
    \begin{itemize}
        \item For $(\rho,\sigma)\in A$, $\Tr\squ*{\Phi(\rho)\Phi(\sigma)}\leq p_0$, and if $\Tr\squ*{\Phi(\rho)\Phi(\sigma)}\geq p_0-\varepsilon$ then there exist $(\rho_0,\sigma_0)\in B$ such that $\norm{\rho-\rho_0}_{\Tr}+\norm{\sigma-\sigma_0}_{\Tr}\leq C\varepsilon^\alpha$ for some constants $C,\alpha>0$ independent of $(\rho,\sigma)$.

        \item For all $(\rho,\sigma)\in B$, $\Tr\squ*{\Psi(\rho)\Psi(\sigma)}\leq q_0$.
    \end{itemize}
    Then, for all $\eta>0$, there exists $p\in(0,1)$ such that the channel $\Phi_p:\mc{B}(H)\rightarrow\mc{B}(K\oplus K')$ defined $\Phi_p=p\Phi\oplus(1-p)\Psi$ satisfies the property that, for all $(\rho,\sigma)\in A$,
    \begin{align}
        \Tr\squ*{\Phi_p(\rho)\Phi_p(\sigma)}\leq p^2p_0+(1-p)^2q_0+\eta.
    \end{align}
\end{lemma}

From the proof below, it may be verified that if $\eta=1/\mathrm{poly}$, $C=\mathrm{poly}$, and $\alpha<1$ is constant, we get a probability $p,1-p=1/\mathrm{poly}$. Further, the construction preserves polynomial gaps. In fact, if we have the promise that $\Tr\squ*{\Psi(\rho)\Psi(\sigma)}\leq s$ for all $(\rho,\sigma)\in B$ or $\Tr\squ*{\Psi(\rho)\Psi(\sigma)}\geq c$ for some $(\rho,\sigma)\in A\cap B$ with $c-s=1/\mathrm{poly}$, then, we have that either $\Tr\squ*{\Phi_p(\rho)\Phi_p(\sigma)}\leq s'$ for all $(\rho,\sigma)\in A$ or $\Tr\squ*{\Phi_p(\rho)\Phi_p(\sigma)}\geq c'$ for some $(\rho,\sigma)\in A\cap B$, where $c'=p^2p_0+(1-p)^2c$ and $s'=p^2p_0+(1-p)^2s+\eta$. So the gap becomes $c'-s'=(1-p)^2(c-s)-\eta$. Since $(1-p)^2=\frac{\eta^{1-\alpha}}{\mathrm{poly}}$, we may choose $\eta=1/\mathrm{poly}$ small enough to preserve the polynomial gap. 

\begin{proof}
    Let $(\rho,\sigma)\in A$. Then, there exists some $\varepsilon\geq 0$ such that $\Tr\squ*{\Phi(\rho)\Phi(\sigma)}=p_0-\varepsilon$. By hypothesis, there exist $(\rho_0,\sigma_0)\in B$ such that $\norm{\rho-\rho_0}_{\Tr}+\norm{\sigma-\sigma_0}_{\Tr}\leq C\varepsilon^\alpha$. Also, we have by hypothesis that $\Tr\squ*{\Psi(\rho_0)\Psi(\sigma_0)}\leq q_0$. Therefore,
    \begin{align}
    \begin{split}
        \Tr\squ*{\Psi(\rho)\Psi(\sigma)}&=\Tr\squ*{\Psi(\rho_0)\Psi(\sigma_0)}+\Tr\squ*{\Psi^\dag(\Psi(\rho_0))(\sigma-\sigma_0)}+\Tr\squ*{(\rho-\rho_0)\Psi^\dag(\Psi(\sigma))}\\
        &\leq\Tr\squ*{\Psi(\rho_0)\Psi(\sigma_0)}+\norm{\sigma-\sigma_0}_{\Tr}+\norm{\rho-\rho_0}_{\Tr}\\
        &\leq q_0+C\varepsilon^\alpha.
    \end{split}
    \end{align}
    Using this bound,
    \begin{align}
    \begin{split}
        \Tr\squ*{\Phi_p(\rho)\Phi_p(\sigma)}&=p^2\Tr\squ*{\Phi(\rho)\Phi(\sigma)}+(1-p)^2\Tr\squ*{\Psi(\rho)\Psi(\sigma)}\\
        &\leq p^2(p_0-\varepsilon)+(1-p)^2(q_0+C\varepsilon^\alpha)\\
        &=p^2p_0+(1-p)^2q_0+(1-p)^2C\varepsilon^\alpha-p^2\varepsilon
    \end{split}
    \end{align}
    To finish the proof, it remains to show that $p$ may be chosen so that $(1-p)^2C\varepsilon^\alpha-p^2\varepsilon\leq\eta$ for all $\varepsilon\geq0$. Without loss of generality, we may assume $\alpha<1$. Then, by extremising the expression with respect to $\varepsilon$, we see that it is maximised at $\varepsilon=\parens*{C\alpha\parens*{\tfrac{1}{p}-1}^2}^{1/(1-\alpha)}$, giving $(1-p)^2C\varepsilon^\alpha-p^2\varepsilon\leq C^{1/(1-\alpha)}\parens*{\alpha^{\alpha/(1-\alpha)}-\alpha^{1/(1-\alpha)}}p^2\parens*{\tfrac{1}{p}-1}^{2/(1-\alpha)}$. Thus, rearranging the wanted expression $ C^{1/(1-\alpha)}\parens*{\alpha^{\alpha/(1-\alpha)}-\alpha^{1/(1-\alpha)}}p^2\parens*{\tfrac{1}{p}-1}^{2/(1-\alpha)}\leq\eta$ gives
    \begin{align}
        \frac{1-p}{p^\alpha}\leq\parens*{\frac{\eta}{C^{1/(1-\alpha)}\parens*{\alpha^{\alpha/(1-\alpha)}-\alpha^{1/(1-\alpha)}}}}^{(1-\alpha)/2}.
    \end{align}
    Since $\frac{1-p}{p}\geq\frac{1-p}{p^\alpha}$, we may take $\frac{1-p}{p}=\parens*{\frac{\eta}{C^{1/(1-\alpha)}\parens*{\alpha^{\alpha/(1-\alpha)}-\alpha^{1/(1-\alpha)}}}}^{(1-\alpha)/2}$, or
    \begin{align}
        p=\frac{1}{1+\parens*{\frac{\eta}{C^{1/(1-\alpha)}\parens*{\alpha^{\alpha/(1-\alpha)}-\alpha^{1/(1-\alpha)}}}}^{(1-\alpha)/2}}
    \end{align}
    to get the result. 
\end{proof}

\begin{lemma}\label{lem:close-to-separable}
    Let $H$ and $K$ be Hilbert spaces and pure states $\rho=\ketbra{\psi},\sigma=\ketbra{\phi}\in \mc{D}(H^{\otimes k}\otimes K)$. Suppose that there exist $\varepsilon_i>0$ such that $\Tr(\rho_i\sigma_i)\geq 1-\varepsilon_i$ for all $i$, where $\rho_i$ and $\sigma_i$ are the marginals on the $i$-th copy of $H$. Write $\varepsilon=\frac{1}{k}\sum_i\varepsilon_i$. Then, there exist separable states $\tilde{\rho}=\ketbra{\psi_1}\otimes\cdots\otimes\ketbra{\psi_k}\otimes\ketbra{\psi'},\tilde{\sigma}=\ketbra{\phi_1}\otimes\cdots\otimes\ketbra{\phi_k}\otimes\ketbra{\phi'}\in 
    \mc{D}(H^{\otimes k}\otimes K)$ such that $\norm{\rho-\tilde{\rho}}_2\leq2\sqrt{k\varepsilon}$, $\norm{\sigma-\tilde{\sigma}}_2\leq2\sqrt{k\varepsilon}$, and $\norm{\Tr_K\tilde{\rho}-\Tr_K\tilde{\sigma}}_2\leq6\sqrt{k\varepsilon}$.
\end{lemma}

\begin{proof}
    Fix some $i=1,\ldots,k$. Then, $\Tr(\rho_i\sigma_i)\geq 1-\varepsilon_i$ implies that $\sqrt{\Tr\rho_i^2}\sqrt{\Tr\sigma_i^2}\geq 1-\varepsilon_i$. In particular, since $\Tr\sigma_i^2\leq 1$, $\Tr\rho_i^2\geq(1-\varepsilon_i)^2$. Now, Let $\rho_i=\sum_j\lambda_{i,j}\ketbra{\psi_{j}^{(i)}}$ be s spectral decomposition of $\rho_i$, with $\lambda_{i,1}\geq\lambda_{i,2}\geq\ldots$. First, the above implies that $\sum_j\lambda_{i,j}^2\geq(1-\varepsilon_i)^2$ and the fact that $\rho_i$ is a state means that $\sum_j\lambda_{i,j}=1$. As such, 
    \begin{align}
        \lambda_{i,1}=\sum_j\lambda_{i,1}\lambda_{i,j}\geq\sum_j\lambda_{i,j}^2\geq(1-\varepsilon_i)^2.
    \end{align}
    Thus, $\braket{\psi_1^{(i)}}{\rho_i}{\psi_1^{(i)}}=\lambda_{i,1}\geq(1-\varepsilon_i)^2$, or equivalently, $\norm{(\Id\otimes\bra{\psi^{(i)}_1}\otimes\Id)\ket{\psi}}\geq1-\varepsilon_i$. Now, we can expand $\ket{\psi}$ in the basis $\ket{\psi^{(1)}_{j_1}}\otimes\cdots\otimes\ket{\psi^{(k)}_{j_k}}$ of $H^{\otimes k}$ as $\ket{\psi}=\sum_{j_1,\ldots,j_k}\ket{\psi_{j_1}}\otimes\cdots\otimes\ket{\psi_{j_k}}\otimes\ket{\psi_{j_1\ldots j_k}}$, where the $\ket{\psi_{j_1\ldots j_k}}\in K$ are subnormalised states such that $\sum_{j_1,\ldots,j_k}\braket{\psi_{j_1\ldots j_k}}=1$. For all $i$, we have that
    \begin{align}
        \sum_{j_i\neq 1}\braket{\psi_{j_1\ldots j_k}}=1-\norm{(\Id\otimes\bra{\psi^{(i)}_1}\otimes\Id)\ket{\psi}}^2\leq 1-(1-\varepsilon_i)^2\geq2\varepsilon_i.
    \end{align}
    Therefore, by a union bound,
    \begin{align}
    \begin{split}
        \braket{\psi_{1\ldots 1}}&=1-\sum_{j_1\neq 1\lor\ldots\lor j_k\neq 1}\braket{\psi_{j_1\ldots j_k}}\geq1-\sum_i\sum_{j_i\neq 1}\braket{\psi_{j_1\ldots j_k}}\\
        &\geq1-2k\varepsilon.
    \end{split}
    \end{align}
    Let $\ket{\psi_i}=\ket{\psi^{(i)}_1}$ and $\ket{\psi'}=\frac{\ket{\psi_{1\ldots 1}}}{\norm{\ket{\psi_{1\ldots 1}}}}$. Then,
    \begin{align}
    \begin{split}
        \norm{\rho-\tilde{\rho}}_2^2&=2-2\norm{(\bra{\psi_1}\otimes\cdots\otimes\bra{\psi_k}\otimes\bra{\psi'})\ket{\psi}}^2=2-2|\braket{\psi'}{\psi_{1\ldots 1}}|^2\\
        &=2-2\braket{\psi_{1\ldots 1}}\leq 4k\varepsilon.
    \end{split}
    \end{align}
    By symmetry, we also get separable $\tilde{\sigma}$ such that $\norm{\sigma-\tilde{\sigma}}_2\leq 2\sqrt{k\varepsilon}$.

    For the last part of the proof, note that
    \begin{align}
    \begin{split}
        \norm{\Tr_K\tilde{\rho}-\Tr_K\tilde{\sigma}}_2^2&=\norm{\ketbra{\psi_1}\otimes\cdots\otimes\ketbra{\psi_k}-\ketbra{\phi_1}\otimes\cdots\otimes\ketbra{\phi_k}}_2^2\\
        &=2-2\abs{\braket{\psi_1}{\phi_1}}^2\cdots\abs{\braket{\psi_k}{\phi_k}}^2\\
        &\hspace{-1cm}=2-2(1-\tfrac{1}{2}\norm{\ketbra{\psi_1}-\ketbra{\phi_1}}_2^2)\cdots(1-\tfrac{1}{2}\norm{\ketbra{\psi_k}-\ketbra{\phi_k}}_2^2)\\
        &\leq\sum_i\norm{\ketbra{\psi_i}-\ketbra{\phi_i}}_2^2.
    \end{split}
    \end{align}
    Next, $\norm{\rho_i-\ketbra{\psi_i}}_2^2=\Tr(\rho_i^2)+1-2\braket{\psi_i}{\rho_i}{\psi_i}\leq 2-2(1-\varepsilon_i)^2\leq 4\varepsilon_i$ and similarly for $\sigma$, so
    \begin{align}
    \begin{split}
        \norm{\ketbra{\psi_i}-\ketbra{\phi_i}}_2&\leq\norm{\ketbra{\psi_i}-\rho_i}_2+\norm{\rho_i-\sigma_i}_2+\norm{\sigma_i-\ketbra{\phi_i}}_2\\
        &\leq2\sqrt{\varepsilon_i}+\sqrt{2\varepsilon_i}+2\sqrt{\varepsilon_i}=(4+\sqrt{2})\sqrt{\varepsilon_i}.
    \end{split}
    \end{align}
    Together, this gives $\norm{\Tr_K\tilde{\rho}-\Tr_K\tilde{\sigma}}_2^2\leq k(4+\sqrt{2})^2\varepsilon$, as wanted
\end{proof}

\begin{lemma}\label{lem:separator}
    Let $H$ and $K$ be Hilbert spaces and $k\in\N$, and define the channel $\Phi_1:\mc{B}(H^{\otimes k}\otimes K)\rightarrow\mc{B}(H\otimes\C^{[k]})$ as
    \begin{align}
        \Phi_1(\rho)=\expec_{i=1}^k\rho_i\otimes\ketbra{i},
    \end{align}
    where $\rho_i$ is the marginal of $\rho$ on the $i$-th copy of $H$ and $\expec$ is the normalised sum (expectation). Let $\rho=\ketbra{\psi},\sigma=\ketbra{\phi}\in\mc{D}(H^{\otimes k}\otimes H')$ be orthogonal pure states. Then, $\Tr\squ*{\Phi_1(\rho)\Phi_1(\sigma)}\leq\frac{1}{k}$ and if $\Tr\squ*{\Phi_1(\rho)\Phi_1(\sigma)}\geq\frac{1}{k}-\varepsilon$, then there exist states $\ket{\psi_1},\ldots,\ket{\psi_k}\in H$ and orthogonal states $\ket{\alpha},\ket{\beta}\in K$ such that
    \begin{align}
        \norm{\rho-\tilde{\rho}}_{\Tr}+\norm{\sigma-\tilde{\sigma}}_{\Tr}\leq 10k\varepsilon^{1/4},
    \end{align}
    where $\tilde{\rho}=\ketbra{\psi_1}\otimes\cdots\otimes\ketbra{\psi_k}\otimes\ketbra{\alpha}$ and $\tilde{\sigma}=\ketbra{\psi_1}\otimes\cdots\otimes\ketbra{\psi_k}\otimes\ketbra{\beta}$.
\end{lemma}

\begin{proof}
    First, $\Tr\squ*{\Phi_1(\rho)\Phi_1(\sigma)}=\frac{1}{k}\expec_{i}\Tr(\rho_i\sigma_i)\leq\frac{1}{k}\expec_i1=\frac{1}{k}$. Now, suppose that $\Tr\squ*{\Phi_1(\rho)\Phi_1(\sigma)}\geq\frac{1}{k}-\varepsilon$. Let $\varepsilon_i\in[0,1]$ such that $\Tr(\rho_i\sigma_i)=1-\varepsilon_i$; we know that $k\varepsilon\geq\expec_i\varepsilon_i$. Then, due to \cref{lem:close-to-separable}, there exist separable states $\rho'=\ketbra{\psi_1}\otimes\cdots\otimes\ketbra{\psi_k}\otimes\ketbra{\alpha},\sigma'=\ketbra{\phi_1}\otimes\cdots\otimes\ketbra{\phi_k}\otimes\ketbra{\phi'}\in\mc{D}(H^{\otimes k}\otimes H')$ such that $\norm{\rho-\rho'}_2\leq2k\sqrt{\varepsilon}$, $\norm{\sigma-\sigma'}_2\leq2k\sqrt{\varepsilon}$, and $\norm{\Tr_{H'}\rho'-\Tr_{H'}\sigma'}_2\leq6k\sqrt{\varepsilon}$. Since all these are pure states, this implies that $\norm{\rho-\rho'}_{\Tr}\leq\sqrt{2}k\sqrt{\varepsilon}$, $\norm{\sigma-\sigma'}_{\Tr}\leq\sqrt{2}k\sqrt{\varepsilon}$, and $\norm{\Tr_{H'}\rho'-\Tr_{H'}\sigma'}_{\Tr}\leq3\sqrt{2}k\sqrt{\varepsilon}$. Let $\sigma''=\Tr_{H'}(\rho')\otimes\ketbra{\phi'}$; then
    \begin{align}
        \norm{\sigma-\sigma''}_{\Tr}\leq\norm{\sigma-\sigma'}_{\Tr}+\norm{\Tr_{H'}\sigma'-\Tr_{H'}\rho'}_{\Tr}\leq 4\sqrt{2}k\sqrt{\varepsilon}.
    \end{align}
    Now, the inner product
    \begin{align}
    \begin{split}
        \abs*{\braket{\alpha}{\phi'}}^2=\Tr\squ*{\rho'\sigma''}\leq\Tr\squ*{\rho\sigma}+\norm{\rho-\rho'}_{\Tr}+\norm{\sigma-\sigma''}_{\Tr}\leq5\sqrt{2}k\sqrt{\varepsilon}.
    \end{split}
    \end{align}
    There exists a state $\ket{\beta}\in H'$ orthogonal to $\ket{\alpha}$ such that $\ket{\phi'}=a\ket{\alpha}+b\ket{\beta}$ for some $|\alpha|^2\leq 5\sqrt{2}k\sqrt{\varepsilon}$. Due to this
    \begin{align}
        \norm{\ketbra{\phi'}-\ketbra{\beta}}_{\Tr}=\frac{1}{\sqrt{2}}\norm{\ketbra{\phi'}-\ketbra{\beta}}_2=\sqrt{\frac{1}{2}\parens*{2-2|b|^2}}=|a|\leq\sqrt{5\sqrt{2}k\sqrt{\varepsilon}}.
    \end{align}
    Let $\tilde{\rho}=\rho'$ and $\tilde{\sigma}=\Tr_{H'}\sigma''\otimes\ketbra{\beta}$. Then we get the wanted result
    \begin{align}
    \begin{split}
        \norm{\rho-\tilde{\rho}}_{\Tr}+\norm{\sigma-\tilde{\sigma}}_{\Tr}&\leq\norm{\rho-\rho'}_{\Tr}+\norm{\sigma-\sigma''}_{\Tr}+\norm{\ketbra{\alpha}+\ketbra{\beta}}_{\Tr}\\
        &\leq5\sqrt{2}k\sqrt{\varepsilon}+\sqrt{5\sqrt{2}k\sqrt{\varepsilon}}\leq 10k\varepsilon^{1/4}.
    \end{split}
    \end{align}
\end{proof}

\begin{lemma}\label{lem:algorithm-channel}
    Let $H$ be a Hilbert space. For any channel $\Phi:\mc{B}(H)\rightarrow\mc{B}(Q)$ define $\Psi:\mc{B}(H\otimes Q)\rightarrow\mc{B}(Q_\perp)$ as
    \begin{align}
        \Psi(\rho)=\braket{0}{(\Phi\otimes\Tr)(\rho)}{0}\mu_Q+\braket{1}{(\Phi\otimes\Tr)(\rho)}{1}\ketbra{\perp}
    \end{align}
    Then, for any $\rho\in\mc{D}(H)$ and $\sigma,\sigma'\in\mc{D}(Q)$, if $\braket{1}{\Phi(\rho)}{1}\geq\frac{3}{4}$, then $\Tr\squ*{\Psi(\rho\otimes\sigma)\Psi(\rho\otimes\sigma')}\geq\frac{19}{32}$; and if $\braket{1}{\Phi(\rho)}{1}\leq\frac{2}{3}$, then $\Tr\squ*{\Psi(\rho\otimes\sigma)\Psi(\rho\otimes\sigma')}\leq\frac{1}{2}$ 
\end{lemma}

\begin{proof}
    We can directly compute this. First, $\Psi(\rho\otimes\sigma)=\braket{0}{\Phi(\rho)}{0}\mu_Q+\braket{1}{\Phi(\rho)}{1}\ketbra{\perp}$. Therefore,
    \begin{align}
    \begin{split}
        \Tr\squ*{\Psi(\rho\otimes\sigma)\Psi(\rho\otimes\sigma')}&=\braket{0}{\Phi(\rho)}{0}^2\Tr(\mu_Q^2)+\braket{1}{\Phi(\rho)}{1}^2|\braket{\perp}|^2\\
        &=\frac{1}{2}\parens*{1-\braket{1}{\Phi(\rho)}{1}}^2+\braket{1}{\Phi(\rho)}{1}^2.
    \end{split}
    \end{align}
    The inequalities follow from this expression.
\end{proof}

\begin{proof}[Proof of \cref{thm:clique-qma2-hard}]
    Let $L=(Y,N)\in\tsf{QMA}(2)_{\frac{3}{4},\frac{1}{4}}$, and let $x\in Y\cup N$. Then, there exists a polynomial-time computable channel $\Phi_x:\mc{B}(Q^{\otimes2p(|x|)})\rightarrow\mc{B}(Q)$ such that, if $x\in Y$, there exist $\rho,\sigma\in\mc{D}(Q^{\otimes p(|x|)})$ such that $\braket{1}{\Phi_x(\rho\otimes\sigma)}{1}\geq \frac{3}{4}$; and if $x\in N$, for all $\rho,\sigma\in\mc{D}(Q^{\otimes p(|x|)})$, $\braket{1}{\Phi_x(\rho\otimes\sigma)}{1}\leq \frac{1}{4}$. Construct the channel $\Psi_x:\mc{B}(Q^{\otimes(2p(|x|)+1)})\rightarrow\mc{B}(Q_\perp)$ from $\Phi_x$ as in \cref{lem:algorithm-channel}. Also, write $H=Q^{\otimes p(|x|)}$ and $K=Q$, and let $\Phi_1:\mc{B}(H\otimes H\otimes K)\rightarrow\mc{B}(H\otimes\C^{[2]})$ as in \cref{lem:separator}. For $p\in(0,1)$, consider the channel $\Phi_p=p\Phi_1\oplus(1-p)\Psi_x$. Using the construction of \cref{def:circuit-direct-sum} and the fact that $\Phi_x$ admits a circuit polynomial in the size of the instance $x$, we see that this channel admits a circuit representation polynomial in $|x|$.
    
    First, suppose that $x\in Y$. Then, I claim that there exist $\rho,\sigma\in\mc{D}(Q^{\otimes(2p(|x|)+1)})$ such that $\Tr\squ*{\Phi_p(\rho)\Phi_p(\sigma)}\geq \frac{p^2}{2}+\frac{19(1-p)^2}{32}$. Since $x\in Y$, there exists a separable state on which $\Phi_x$ accepts with high probability; we may, without loss of generality, assume that this state is pure so of the form $\ket{\psi}\otimes\ket{\phi}$. Let $\rho=\ketbra{\psi}\otimes\ketbra{\phi}\otimes\ketbra{0}$ and $\sigma=\ketbra{\psi}\otimes\ketbra{\phi}\otimes\ketbra{1}$. Then, $\Phi_p(\rho)=\frac{p}{2}\parens*{\ketbra{\psi}\otimes\ketbra{1}+\ketbra{\phi}\otimes\ketbra{2}}\oplus (1-p)\Psi_x(\rho)$, and similarly for $\sigma$, giving that
    \begin{align}
    \begin{split}
        \Tr\squ*{\Phi_p(\rho)\Phi_p(\sigma)}&=\frac{p^2}{4}\parens*{|\braket{\psi}|^2+|\braket{\phi}|^2}+(1-p)^2\Tr\squ*{\Psi_x(\rho)\Psi_x(\sigma)}\\
        &\geq \frac{p^2}{2}+(1-p)^2\frac{19}{32}.
    \end{split}
    \end{align}
    
    Now, suppose that $x\in N$. We use \cref{lem:like-5.5} to show that for all $\eta>0$, there exists $p\in (0,1)$ such that $\Tr\squ*{\Phi_p(\rho)\Phi_p(\sigma)}\leq\frac{p^2}{2}+\frac{(1-p)^2}{2}+\eta$. We need that $\Phi_1$ and $\Psi_x$ satisfy the conditions of lemma. Let
    \begin{align}
        &A=\set*{(\ketbra{\psi},\ketbra{\phi})}{\ketbra{\psi},\ketbra{\phi}\in\mc{D}(Q^{\otimes (2p(|x|)+1)}),\;\braket{\psi}{\phi}=0},\\
        \begin{split}
        &B=\big\{(\ketbra{\psi_1}\otimes\ketbra{\psi_2}\otimes\ketbra{\alpha},\ketbra{\psi_1}\otimes\ketbra{\psi_2}\otimes\ketbra{\beta})\\
        &\qquad\qquad\big|\ketbra{\psi_1},\ketbra{\psi_2}\in\mc{D}(Q^{\otimes p(|x|)}),\;\ketbra{\alpha},\ketbra{\beta}\in\mc{D}(Q),\;\braket{\alpha}{\beta}=0\big\}.
        \end{split}
    \end{align}
    First, due to \cref{lem:separator}, for all $(\rho,\sigma)\in A$, $\Tr\squ*{\Phi_1(\rho)\Phi_1(\sigma)}\leq\frac{1}{2}$, and if $\Tr\squ*{\Phi_1(\rho)\Phi_1(\sigma)}\geq\frac{1}{2}-\varepsilon$, then there exists $(\tilde{\rho},\tilde{\sigma})\in B$ such that $\norm{\rho-\tilde{\rho}}_{\Tr}+\norm{\sigma-\tilde{\sigma}}_{\Tr}\leq 20\varepsilon^{1/4}$. Also, due to \cref{lem:algorithm-channel}, we know that for all $(\tilde{\rho},\tilde{\sigma})\in B$, $\Tr\squ*{\Psi_x(\tilde{\rho})\Psi_x(\tilde{\sigma})}\leq \frac{1}{2}$. As such, \cref{lem:like-5.5} gives that for all $\eta>0$, there exists $p\in (0,1)$ such that $\Tr\squ*{\Phi_p(\rho)\Phi_p(\sigma)}\leq\frac{p^2}{2}+\frac{(1-p)^2}{2}+\eta$ for all $(\rho,\sigma)\in A$. As this bound holds for all pure states, it holds for all states. This gives that $\ttt{qClique}(2)_{p^2/2+19(1-p)^2/32,p^2/2+(1-p)^2/2+\eta}$ is $\tsf{QMA}(2)_{c,s}$-hard. Finally, as remarked near~\cref{lem:like-5.5}, $\eta$ may be chosen so that $p^2/2+19(1-p)^2/32$ and $p^2/2+(1-p)^2/2+\eta$ have constant gap $3(1-p)^2/32-\eta$, as the coefficient from \cref{lem:separator} is constant in the case $k=2$.
\end{proof}

\subsection{Reducing \textsf{QMA}(\textit{k}) to \textsf{QMA}(2) using quantum cliques}\label{subsec:NewReduction}

In this section, we give an alternate proof of $\tsf{QMA}(k)=\tsf{QMA}(2)$ (for $k$ polynomial) than the original proof of Harrow and Montanaro~\cite{HM13}. To do this, we show that $\ttt{qClique}(2)_{c,s}$ with some polynomial gap is hard for $\tsf{QMA}(k)$, independently of the result of~\cite{HM13}. The proof is similar to the case $k=2$ that we showed in \cref{thm:clique-qma2-hard}. This implies the wanted equality via \cref{lem:in-qma2}, which showed that $\ttt{qClique}(2)_{c,s}\in\tsf{QMA}(2)$ for all $c,s$ with polynomial gap.

    \begin{theorem}\label{thm:clique-qmak-hard}
        For all $c,s:\N\rightarrow(0,1)_{\exp}$ with polynomial gap and $k:\N\rightarrow\N$ polynomial, there exist $c',s':\N\rightarrow(0,1)_{\exp}$ with polynomial gap such that $\ttt{qClique}(2)_{c',s'}$ is $\tsf{QMA}(k)_{c,s}$-hard.
    \end{theorem}

    \begin{proof}
         The proof proceeds analogously to \cref{thm:clique-qma2-hard}. Let $L=(Y,N)\in\tsf{QMA}(k)_{c,s}$. By increasing the number of proofs, it is possible to amplify the probability gap. In particular, there exists a polynomial $k':\N\rightarrow\N$ such that $L\in\tsf{QMA}(k')_{\frac{3}{4},\frac{1}{4}}$. Let $x\in Y\cup N$ be an instance of $L$. Now, let $\Phi_x:\mc{B}(Q^{\otimes k'p(|x|)})\rightarrow\mc{B}(Q)$ be the channel of the $\tsf{QMA}(k')_{\frac{3}{4},\frac{1}{4}}$ verifier for $L$, from which we can follow the construction of \cref{lem:algorithm-channel} to get $\Psi_x:\mc{B}(Q^{\otimes k'p(|x|)+1})\rightarrow\mc{B}(Q_{\perp})$. Also, let  $H=Q^{\otimes p(|x|)}$, $K=Q$, and $\Phi_1:\mc{B}(H^{\otimes k'}\otimes K)\rightarrow\mc{B}(H\otimes\C^{[k']})$ be the channel from \cref{lem:separator} with $k=k'$. Let $\Phi_p=p\Phi_1\oplus(1-p)\Psi_x$; we consider the cliques of this channel. Note that, using the construction of \cref{def:circuit-direct-sum}, it is direct to see that this channel admits a circuit representation polynomial in the size of the instance $x$.

         First, suppose $x\in Y$. Then, there exist some separable state $\rho_0=\ketbra{\psi_1}\otimes\cdots\otimes\ketbra{\psi_k}$ such that $\braket{1}{\Phi_x(\rho_0)}{1}\geq\frac{3}{4}$. As such, we have that with $\rho=\rho_0\otimes\ketbra{0}$ and $\sigma=\rho_0\otimes\ketbra{1}$, $\Tr\squ*{\Phi_p(\rho)\Phi_p(\sigma)}\geq\frac{p^2}{k'}+(1-p)^2\frac{19}{32}$, giving that $(\rho,\sigma)$ forms a $(\frac{p^2}{k'}+(1-p)^2\frac{19}{32},2)$-clique. On the other hand, if $x\in N$, then, we know that for all separable $\rho_0\in(Q^{\otimes p(|x|)})^{\otimes k'}$, $\braket{1}{\Phi_x(\rho_0)}{1}\leq\frac{1}{4}$. As such, we know by \cref{lem:like-5.5} that for all $\eta>0$, there exists $p$ such that $\Phi_p$ has no $\parens*{\frac{p^2}{k'}+\frac{(1-p)^2}{2}+\eta,2}$-clique, as in \cref{thm:clique-qma2-hard}. This gives a probability gap $(1-p)^2\frac{3}{32}-\eta$, which can be chosen to be inverse polynomial, as the coefficient from \cref{lem:separator}, $10k'$, is polynomial.
    \end{proof}
    
    \begin{corollary}
        For all $c,s:\N\rightarrow(0,1)_{\exp}$ with polynomial gap and $k:\N\rightarrow\N$ polynomial, there exist $c',s':\N\rightarrow(0,1)_{\exp}$ with polynomial gap such that $\tsf{QMA}(k)_{c,s}\subseteq\tsf{QMA}(2)_{c',s'}$.
    \end{corollary}

    The corollary follows immediately by concatenating the results of \cref{thm:clique-qmak-hard} and \cref{lem:in-qma2}.

\subsection{What about quantum independent sets?}\label{sec:qis}

The quantum independent set problem, defined in \cref{def:qClique/qIS,def:qcliqueqISProblems}, seems very similar to the quantum clique problem. However, its exact complexity remains vexing. Nevertheless, we are able to show that, as for the clique problem, the independent set problem is contained in $\tsf{QMA}(2)$.

\begin{lemma}
    Let $c,s:\N\rightarrow[0,1]$ and $k:\N\rightarrow\N$. Then, $\ttt{qIS}(k)_{c,s}\in\tsf{QMA}(k)_{c',s'}$ for $c'=\frac{1+\frac{c-s}{2}c}{2+c-s}$ and $s'=\frac{1+\frac{c-s}{2}\frac{c+s}{2}}{2+c-s}$.
\end{lemma}

This implies immediately that for $k$ polynomial and $c,s:\N\rightarrow(0,1)_{\exp}$ with polynomial gap $\ttt{qIS}(k)_{c,s}\in\tsf{QMA}(2)$, since the gap $c'-s'=\frac{(c-s)^2}{4(2+c-s)}\geq\frac{(c-s)^2}{16}$ remains polynomial. The proof proceeds in the same way as the proof of \cref{lem:in-qma2}.

\begin{proof}
    Suppose the verifier receives an instance $C\in Y\cup N$ and a $\tsf{QMA}(k)$ proof $\rho_1\otimes\cdots\otimes\rho_k$. Let $\Phi=\Phi_C$. Then, the verifier effects the following procedure.
    \begin{enumerate}[1.]
        \item The verifier samples a uniformly random pair of distinct $i,j\in[k]$ and a random bit $b$, which is $0$ with probability $p$.

        \item If $b=0$, the verifier runs the swap test on $\rho_i\otimes \rho_j$ and outputs $1$ if and only if it fails.

        \item If $b=1$, the verifier computes $\Phi(\rho_i)\otimes\Phi(\rho_j)$, and again runs the swap test, outputting $1$ if it fails.
    \end{enumerate}
    Suppose $C\in Y$; then there exists a $(c(|C|),k)$-independent $\rho_1,\ldots,\rho_k$ set of $\Phi$. So, the prover can send $\rho_1\otimes\cdots\otimes\rho_k$. If $b=0$, the probability of accepting is $\expec_{i,j}\frac{1}{2}-\frac{1}{2}\Tr(\rho_i\rho_j)=\frac{1}{2}$; if $b=1$, the probability of accepting is $\expec_{i,j}\frac{1}{2}-\frac{1}{2}\Tr(\rho_i\rho_j)\geq\frac{1}{2}-\frac{1-c(|C|)}{2}=\frac{1}{2}c(|C|)$. Thus, the probability of accepting is $c'(|C|)=\frac{p+(1-p)c(|C|)}{2}$.

    Now, suppose $C\in N$. We may suppose that each $\rho_i=\ketbra{\psi_i}$ is pure. If $b=0$, then the probability of accepting is $\frac{1}{2}-\frac{1}{2}\expec_{i\neq j}|\braket{\psi_i}{\psi_j}|^2$; if $b=1$, the probability of accepting is $\frac{1}{2}-\frac{1}{2}\expec_{i\neq j}\Tr\squ*{\Phi(\rho_i)\Phi(\rho_j)}$. Let $\Psi$ be the matrix whose columns are $\ket{\psi_i}$, and let $\Psi=U\Sigma V^\dag$ be its singular value decomposition. The matrix with orthonormal columns $\Psi'$ nearest to $\Psi$ in Frobenius norm is $\Psi'=UV^\dag$ --- this is the solution to the orthogonal Procrustes problem~\cite{Sch66}, see also \cref{lem:orthogonaliser} for more. The norm distance between these matrices is
    \begin{align}
        \norm{\Psi-\Psi'}_2^2=\norm{\Sigma-\Id}_2^2=\sum_i(\sigma_i-1)^2\leq\sum_i(\sigma_i^2-1)^2=\norm{\Psi^\dag\Psi-\Id}^2_2=\sum_{i\neq j}\abs*{\braket{\psi_i}{\psi_j}}^2.
    \end{align}
    Let $\ket{\psi_i'}$ be the columns of $\Psi'$. Then, as these vectors are orthonormal, the fact there is no independent set implies $\expec_{i\neq j}\Tr\squ*{\Phi(\ketbra{\psi_i'})\Phi(\ketbra{\psi_j'})}\geq 1-s(|C|)$. As such, we can upper bound the probability of accepting if $b=1$ by
    \begin{align}
        \frac{1}{2}-\frac{1}{2}\expec_{i\neq j}\Tr\squ*{\Phi(\rho_i)\Phi(\rho_j)}&\leq\frac{1}{2}-\frac{1}{2}\expec_{i\neq j}\Tr\squ*{\Phi(\ketbra{\psi_i'})\Phi(\ketbra{\psi_j'})}+\expec_{i}\norm{\ketbra{\psi_i}-\ketbra{\psi_i'}}_{\Tr}\nonumber\\
        \begin{split}
        &\leq\frac{1}{2}-\frac{1-s(|C|)}{2}+\expec_i\norm{\ket{\psi_i}-\ket{\psi_i'}}\\
        &\leq\frac{s(|C|)}{2}+\sqrt{\expec_i\norm{\ket{\psi_i}-\ket{\psi_i'}}}=\frac{s(|C|)}{2}+\norm{\Psi-\Psi'}_2,
        \end{split}
    \end{align}
    using Jensen's inequality. Thus, putting the two cases together, the probability of accepting is
    \begin{align}
    \begin{split}
        s'(|C|)&\leq p\parens*{\frac{1}{2}-\frac{1}{2}\norm{\Psi-\Psi'}_2^2}+(1-p)\parens*{\frac{s(|C|)}{2}+\norm{\Psi-\Psi'}_2}\\
        &=\frac{p+(1-p)s(|C|)}{2}+(1-p)\norm{\Psi-\Psi'}_2-\frac{p}{2}\norm{\Psi-\Psi'}_2^2\\
        &\leq \frac{p+(1-p)s(|C|)}{2}+\frac{(1-p)^2}{2p},
    \end{split}
    \end{align}
    by maximising over $\norm{\Psi-\Psi'}_2\in\R$. Taking $p=\frac{2}{2+c-s}$ gives the wanted values of $c'$ and $s'$.
\end{proof}

\section{Specifying to \textsf{QMA}}\label{Sec:QMA}

Many $\tsf{NP}$-complete computational problems quantise naturally to $\tsf{QMA}$-complete. We have seen, however, that the clique problem quantises rather to a $\tsf{QMA}(2)$-complete problem. This situation begs the question of whether there is some subclass of circuits on which the clique problem is sufficiently weakened to be $\tsf{QMA}$-complete. In this section we show that these circuits are those representing (a large subclass of) the entanglement-breaking channels, consisting of a measurement followed by state preparation. This contrasts well with the result of the previous section, wherein we showed that the clique problem for direct sums of the entanglement-breaking channels with partial traces end up being $\tsf{QMA}(2)$-complete. First, we recall the definition of the class $\tsf{QMA}$.

\begin{definition}
    Let $c,s:\N\rightarrow(0,1)$. A promise problem $Y,N\subseteq\{0,1\}^\ast$ is in $\tsf{QMA}_{c,s}$ if there exist polynomials $p,q:\N\rightarrow\N$ and a Turing machine $V$ with one input tape and one output tape such that
    \begin{itemize}
        \item For all $x\in\{0,1\}^\ast$, $V$ halts on input $x$ in $q(|x|)$ steps and outputs the description of a quantum circuit from $p(|x|)$ qubits to one qubit.

        \item For all $x\in Y$, there exists a state $\rho\in\mc{D}(Q^{\otimes p(|x|)})$ such that $\braket{1}{\Phi_{V(x)}(\rho)}{1}\geq c(|x|)$.

        \item For all $x\in N$ and $\rho\in\mc{D}(Q^{\otimes p(|x|)})$, $\braket{1}{\Phi_{V(x)}(\rho)}{1}\leq s(|x|)$.
    \end{itemize}
\end{definition}

As long as $c,s:\N\rightarrow(0,1)_{\exp}$ with polynomial gap, we have that $\tsf{QMA}_{c,s}=\tsf{QMA}_{\frac{2}{3},\frac{1}{3}}$, so as previously we take this to be the class $\tsf{QMA}$.

\begin{definition}
    Write $\mc{C}_{\mathrm{EB}'}\subseteq\mc{C}_{\mathrm{EB}}$ for those circuits whose output is independent of at least one of the input qubits, \emph{i.e.} there exists a qubit $R$ such that $\Phi_C(\rho)=\Phi_C(\Tr_R(\rho)\otimes\sigma_R)$ for some $\sigma\in\mc{D}(Q)$.
\end{definition}

We aim to first show that $\ttt{qClique}(2,\mc{C}_{\mathrm{EB}'})\in\tsf{QMA}$. Note that we are only showing inclusion of a subclass of the entanglement-breaking channels in $\tsf{QMA}$. Why not all of them? For general channels, our proof of inclusion of the clique problem in $\tsf{QMA}(2)$ has the verifier run the swap test to verifiy orthogonality of the clique. However, in $\tsf{QMA}$, the verifier does not have this power --- in fact, Harrow and Montanaro show that there is no LOCC swap test~\cite{HM13}. As such, we need some other way that the verifier can guarantee that the provided clique is orthogonal. Leaving the result independent of one of the qubits allows the orthogonality to be provided by that qubit. In fact, the prover need not even provide orthogonal states to the verifier for a yes instance: the verifier checks on those qubits that are measured and for the remaining one she may be certain that there are extensions of the given states to that qubit that are both orhogonal and a clique.

We use a result of Brand\~ao~\cite{Bra08} to prove the inclusion. First, we need to define a useful subclass of $\tsf{QMA}(2)$. A two-answer POVM $\{M,\Id-M\}$ is a Bell measurement if there exist some $m,n\in\N$, $p_{ij}\in[0,1]$, and POVMs $\{M_i\}_{i=1}^m$, $\{N_j\}_{j=1}^n$ such that
\begin{align}
    M=\sum_{i=1}^m\sum_{j=1}^n p_{ij}M_i\otimes N_j.
\end{align}
Let $\tsf{QMA}^{\mathrm{Bell}}(2)\subseteq\tsf{QMA}(2)$ be the subset of problems where the verifier may make a Bell measurement in order to decide whether to accept or reject.

\begin{theorem}[Special case of theorem 6.5.2 in \cite{Bra08}]
    $\tsf{QMA}^{\mathrm{Bell}}(2)=\tsf{QMA}$.
\end{theorem}

Now, we use this to show the inclusion.

\begin{proposition}\label{prop:eb-bell}
    Let $c,s:\N\rightarrow[0,1]$. Then, $\ttt{qClique}(2,\mc{C}_{\mathrm{EB}'})_{c,s}\subseteq\tsf{QMA}^{\mathrm{Bell}}(2)_{\frac{1+c}{2},\frac{1+s}{2}}$.
\end{proposition}

Using the result of Brand\~ao, we get as a corollary that $\ttt{qClique}(2,\mc{C}_{\mathrm{EB}'})_{c,s}\in\tsf{QMA}$ when $c,s:\N\rightarrow(0,1)_{\exp}$ with polynomial gap.

\begin{proof}
    Note that if $\Phi$ is a channel whose output is independent of one of its qubits $Q_i$, then if there exists a pair of states $\rho,\sigma$ such that $\Tr\squ*{\Phi(\rho)\Phi(\sigma)}\geq\alpha$, then there exists a pair of orthogonal states $\rho'=\Tr_{Q_i}(\rho)\otimes\ketbra{0}$, $\sigma'=\Tr_{Q_i}(\sigma)\otimes\ketbra{1}$ such that $\Tr\squ*{\Phi(\rho')\Phi(\sigma')}=\Tr\squ*{\Phi(\rho)\Phi(\sigma)}\geq\alpha$. As such, the verifier need not verify that the two states in the clique are orthogonal. Thus, given some circuit $C$ and a purported proof $\rho\otimes\sigma$, the verifier simply computes $\Phi(\rho)\otimes\Phi(\sigma)$, performs the swap test, and outputs $1$ if and only if the swap test accepts. If $C$ is a yes instance, then the prover can provide a proof $\rho\otimes\sigma$ such that $\Tr\squ*{\Phi(\rho)\Phi(\sigma)}\geq c(|C|)$, in which case the prover accepts with probability $\frac{1+c(|C|)}{2}$. Otherwise, if $C$ is a no instance, then for all states $\rho,\sigma$, $\Tr\squ*{\Phi(\rho)\Phi(\sigma)}\leq s(|C|)$, by the remark at the beginning of the proof. As such, the swap test accepts with probability $\frac{1+s(|C|)}{2}$.

    It remains to show that the measurement that the verifier makes is a Bell measurement. The measurement operator for outcome $1$ is $M=(\Phi^\dag\otimes\Phi^\dag)(\Pi_{sym})$ where $\Pi_{sym}$ is the projector onto the symmetric subspace. Writing $\Phi(\rho)=\sum_i \Tr\squ*{M_i\rho}\sigma_i$, as it is entanglement-breaking, we have that $\Phi^\dag(A)=\sum_i\Tr\squ*{A\sigma_i}M_i$, so
    \begin{align}
        M=\sum_{i,j}\Tr\squ*{\Pi_{sym}(\sigma_i\otimes\sigma_j)}M_i\otimes M_j,
    \end{align}
    which is manifestly a Bell measurement operator.
\end{proof}

\begin{theorem}\label{thm:qma-hard}
    There exist $c,s:\N\rightarrow(0,1)_{\exp}$ with polynomial gap such that $\ttt{qClique}(2,\mc{C}_{\mathrm{EB}'})_{c,s}$ is $\tsf{QMA}$-hard.
\end{theorem}

\begin{proof}
    We show this by reducing the local Hamiltonian problem to the clique problem. Let $H=\sum_{i=1}^t H_i\in\mc{B}(Q^{\otimes n})$ be an instance of the $k$-local Hamiltonian problem with completeness $\alpha$ and soundness $\beta$. By construction~\cite{KSV02}, we may assume that the Hamiltonian is normalised in the sense that $0\leq H_i\leq\Id$. By gap amplification, we may further assume that $\beta>4\alpha$ with polynomial gap. Let $c=1-2\alpha/t$ and $s=1-\beta/(2t)$. This gives $c$ and $s$ with polynomial gap.
    
    We construct the channel $\Phi:\mc{B}(Q^{\otimes n+1})\rightarrow\mc{B}(Q_\perp)$ as
    \begin{align}
    \begin{split}
        \Phi(\rho)&=\Tr\squ*{(H/t\otimes\Id)\rho}\mu_Q+\Tr\squ*{\parens*{(\Id-H/t)\otimes\Id}\rho}\ketbra{\perp}\\
        &=\frac{1}{t}\sum_{i=1}^t\Tr\squ*{(H_i\otimes\Id)\rho}\mu_Q+\Tr\squ*{\parens*{(\Id-H_i)\otimes\Id}\rho}\ketbra{\perp}.
    \end{split}
    \end{align}
    Since this channel consists of randomly sampling $i\in\{1,\ldots,t\}$ uniformly at random, measuring according to the POVM $\{H_i,\Id-H_i\}$, and then preparing a state of fixed size, it can be described by a circuit $C_H$ of size polynomial in the size of the description of $H$. 
    
    Now, suppose $H$ is a yes instance. Then, there exists a state $\braket{\psi}{H}{\psi}\leq\alpha$. So, taking $\rho=\ketbra{\psi}\otimes\ketbra{0}$, $\sigma=\ketbra{\psi}\otimes\ketbra{1}$, we get that
    \begin{align}
    \begin{split}
        \Tr\squ*{\Phi(\rho)\Phi(\sigma)}&=\frac{1}{2}(\braket{\psi}{H}{\psi}/t)^2+(1-\braket{\psi}{H}{\psi}/t)^2\\
        &\geq1-2\braket{\psi}{H}{\psi}/t\geq 1-2\alpha/t=c.
    \end{split}
    \end{align}
    On the other hand, if $H$ is a no instance, then for all states $\rho\in\mc{B}(Q^{\otimes n})$, $\Tr\squ*{H\rho}\geq\beta$. Then, for any state $\rho\in\mc{B}(Q^{\otimes n+1})$, we have that $\Tr\squ*{(H\otimes\Id)\rho}\geq\beta$, giving
    \begin{align}
    \begin{split}
        \Tr\squ*{\Phi(\rho)^2}&=\frac{1}{2}\parens*{\Tr\squ*{(H\otimes\Id)\rho}/t}^2+(1-\Tr\squ*{(H\otimes\Id)\rho}/t)^2\\
        &\leq1-\frac{1}{2}\Tr\squ*{(H\otimes\Id)\rho}/t\leq1-\beta/2t=s.
    \end{split}
    \end{align}
    Thus, for any pair of orthogonal states $\rho,\sigma$, the triangle inequality gives that $\Tr\squ*{\Phi(\rho)\Phi(\sigma)}\leq\sqrt{\Tr\squ*{\Phi(\rho)^2}\Tr\squ*{\Phi(\sigma)^2}}\leq s$.
\end{proof}

We finish this section by noting what happens to the quantum $2$-independent set problem on this reduced set of circuits. It turns out it is also a $\tsf{QMA}$-complete problem in an almost identical way to the clique problem, so we only provide a sketch of this here. It is an easy exercise to verify that \cref{prop:eb-bell} may be slightly modified to show that $\ttt{qIS}(2,\mc{C}_{\mathrm{EB}'})_{c,s}\subseteq\tsf{QMA}^{\mathrm{Bell}}(2)_{\frac{c}{2},\frac{s}{2}}$, and hence the quantum $2$-independent set problem for this subset of entanglement-breaking channels is in $\tsf{QMA}$. On the other hand, by a construction similar to \cref{thm:qma-hard} and to a construction in Theorem 3.1 of~\cite{BS08}, we also get that the problem is $\tsf{QMA}$-hard. As such, this leads to the following theorem.

\begin{theorem}
    There exist $c,s:\N\rightarrow(0,1)_{\exp}$ with polynomial gap such that $\ttt{qIS}(2,\mc{C}_{\mathrm{EB}'})_{c,s}$ is $\tsf{QMA}$-complete.
\end{theorem}

\section{Classical Collision and Clique Problems}\label{Sec:ClassicalCase}

In this section, we formally introduce the classical versions of the clique and collisions problems and study their computational complexity.

\subsection{The deterministic case}\label{subsec:DeterministicCase}

The \emph{collision problem} for functions is of fundamental interest to complexity theory and cryptography. It asks whether, given some function $f:X\rightarrow Y$, is there a pair of elements in the domain $x,y\in X$ with the same image $f(x)=f(y)$ --- a collision.

Collisions also have a natural interpretation in the language of channel capacity. The function $f$ can be seen as a deterministic (\emph{i.e.} not noisy) channel from the input set $X$ to the output set $Y$. As in the case for a noisy channel, this induces a \emph{confusability graph} on the set $X$: $x$ and $y$ are connected by an edge if and only if $f(x)=f(y)$. That is, the confusability graph indicates the collisions of the function. The graphs that this procedure induce are, in themselves, not particularly interesting --- they are simply disjoint unions of complete graphs on the subsets of the input set that collide --- but they allow to interpret collisions in a graph-theoretic way, as cliques. With this in mind, we define the following graph-theoretic notions for functions.

\begin{definition} Let $X,Y$ be sets and $k\in\N$.
    \begin{itemize}
        \item A function $f:X\rightarrow Y$ has a \emph{$k$-clique} (or $k$-collision) if there exist distinct $x_1,\ldots,x_k\in X$ such that $f(x_1)=\ldots=f(x_k)$.

        \item A function $f:X\rightarrow Y$ has a $k$-independent set if there exist distinct $x_1,\ldots,x_k\in X$ such that $f(x_i)\neq f(x_j)$ for all $i\neq j$.
    \end{itemize}
\end{definition}

To study the complexity of finding cliques and independent sets, we fix a presentation for the functions we consider. The most basic way to present a function is as a table of values. Though this is the most general, it is highly inefficient: a function that may be easy to compute is described with the same complexity as function that is more difficult, and the length of the function description is exponential in the size of the input set. Two natural and equivalent options from the point of view of computational complexity are to present a function as a Turing machine or as a circuit diagram. We will use the circuit description, as this will generalise more naturally to quantum channels, represented as quantum circuits. With this choice of convention, we can formally present the clique and independent set problems in the deterministic case.

\begin{definition}
    Let $k\in\N$. The language for the \emph{$k$-clique problem} is
    \begin{align}
        \ttt{Clique}(k)=\set*{\text{circuits }C}{f_C:\{0,1\}^{\ttt{in}(C)}\rightarrow\{0,1\}^{\ttt{out}(C)}\text{ has a $k$-clique}}.
    \end{align}
    The language for the \emph{$k$-independent set problem} is
    \begin{align}
        \ttt{IS}(k)=\set*{\text{circuits }C}{f_C:\{0,1\}^{\ttt{in}(C)}\rightarrow\{0,1\}^{\ttt{out}(C)}\text{ has a $k$-independent set}}.
    \end{align}
\end{definition}

It is well known that the collision problem for general functions is $\tsf{NP}$-complete; we will give simple constructions that provide this completeness in the context of our graph-theoretic language. First, we recall the definition of $\tsf{NP}$.

\begin{definition}
    A language $L\subseteq\{0,1\}^\ast$ is in $\tsf{NP}$ if there exist polynomials $p,q:\N\rightarrow\N$ and a Turing machine $V$ with two input tapes and one output tape such that
    \begin{itemize}
        \item For all $x,y\in\{0,1\}^\ast$, $V$ halts on input $(x,y)$ in $q(|x|)$ steps.

        \item For all $x\in L$, there exists $y\in\{0,1\}^{p(|x|)}$ such that $V(x,y)=1$.

        \item For all $x\notin L$ and for all $y\in\{0,1\}^\ast$, $V(x,y)=0$.
    \end{itemize}
\end{definition}

To finish this section, we give proofs of the $\tsf{NP}$-completeness of both the $2$-clique and the $2$-independent set problems. Note that the proofs of inclusion of the $2$-clique and $2$-independent set problems in $\tsf{NP}$ generalise directly to the case of $k$-cliques and $k$-independent sets (for $k$ polynomial in the instance size), showing that these \emph{a priori} harder problems are also $\tsf{NP}$-complete. In \cref{sec:a-k-cliques}, we also provide a direct reductions to the $2$-clique and independent set problems. 

\begin{proposition}
    $\ttt{Clique}(2)$ is $\tsf{NP}$-complete.
\end{proposition}

\begin{proof}
    First, it is direct to see that the problem $\ttt{Clique}(2)\in\tsf{NP}$. In fact, given a circuit $C$, the verifier expects to receive a proof of the form $(x,y)$ for $x,y\in\{0,1\}^{\ttt{in}(C)}$, computes $f_C(x)$ and $f_C(y)$, and accepts if and only if they are equal. If $C$ is a yes instance, then the prover can simply provide a collision, causing the verifer to accept; and else, there is no clique for the prover to provide, so the verifier never accepts.

    Next, we need to see that $\ttt{Clique}(2)$ is $\tsf{NP}$-hard. Let $L\in\tsf{NP}$, let $V:\{0,1\}^\ast\times\{0,1\}^\ast\rightarrow\{0,1\}$ be its verifier Turing machine, $q$ be the halting time, and $p$ be the proof length. Then, for each $x\in\{0,1\}^\ast$, let $f_x:\{0,1\}^{p(|x|)}\times\{0,1\}\rightarrow\{0,1\}^{p(|x|)}\times\{0,1\}$ be the function $f_x(y,b)=(y,b\lor V(x,y))$. Since $V(x,y)$ halts in $q(|x|)$ steps, $f_x$ is computable in time polynomial in $|x|$, and hence has a circuit description $C_x$ that is polynomial size in $|x|$. This provides an instance of the collision problem. Now, if $x\in L$, there exists $y$ such that $V(x,y)=1$, so $f_x(y,0)=(y,1)=f_x(y,1)$, so $C_x\in\ttt{Clique}(2)$; on the other hand, if $x\notin L$, then for all $y$, $V(x,y)=0$, giving $f_x(y,b)=(y,b)$ is injective, so $C_x\notin\ttt{Clique}(2)$. As such, we have the wanted Karp reduction.
\end{proof}

\begin{proposition}
    $\ttt{IS}(2)$ is $\tsf{NP}$-complete.
\end{proposition}

\begin{proof}
    First, we need that $\ttt{IS}(2)\in\tsf{NP}$. As for the clique problem, the proof for a circuit $C$ can be taken to be the description of the independent set $(x,y)$. In that case, the verifier simply needs to compute $f_C(x)$ and $f_C(y)$ and check that they are different, something that it can do efficiently in the length of the circuit. 

    Conversely, let $L\in\tsf{NP}$ be some language and let $V$ be a verifier for $L$. Now for $x\in\{0,1\}^\ast$, consider the function $f_x:\{0,1\}^{p(|x|)}\rightarrow\{0,1\}^2$ defined as $f_x(y)=(y_1\land V(x,y),\lnot y_1\land V(x,y))$, where $y_1$ is the first bit of $y$. If $x\notin L$, then $V(x,y)=0$ for all $y\in\{0,1\}^{p(|x|)}$, so $f_x(y)=(0,0)$ giving that there is no independent set. On the other hand, if there is $y$ such that $V(x,y)=1$, $f_x(y)=(y_1,\lnot y_1)\neq(0,0)$. Then, for any $z\in\{0,1\}^{p(|x|)}$ such that $z_1\neq y_1$, $f_x(y)\neq f_x(z)$. So, $f_x$ has a $2$-independent set if and only if $x\in L$. Finally, it is direct to see that, since the circuit of $y\mapsto V(x,y)$ is efficiently describable in $|x|$, the circuit of $f_x$ is also.
\end{proof}

\subsection{The probabilistic case}\label{subsec:Probabilistic}

The results of the previous section generalise easily to probabilistic functions. This is the setting where cliques and independent sets of functions were originally considered, since probabilistic functions model noisy channels~\cite{Sha56}. First, we formally introduce probabilistic functions.

\begin{definition}
    A \emph{probabilistic function} $f:X\rightarrow Y$ is a function that to each $x\in X$ associates a random variable $f(x)$ on $Y$.
\end{definition}

Intuitively, cliques of a noisy channel describe inputs that are likely to become confused, and independent sets describe inputs that are unlikely to. We can think of this in the context of zero-error capacity, wherein cliques are those inputs that are always confused, and independent sets those that are never confused. However, from a complexity-theoretic point of view, this is too restrictive, since we are looking to extract a property of a probabilistic object that should hold exactly. To remedy this, we define cliques/independent sets with respect to an additional parameter that describes the probability of confusion.

\begin{definition} Let $X,Y$ be sets, $k\in\N$, and $\alpha\in[0,1]$.
    \begin{itemize}
        \item A probabilistic function $f:X\rightarrow Y$ has an \emph{$(\alpha,k)$-clique} if there exist distinct $x_1,\ldots,x_k\in X$ such that for all $i\neq j$,
        \begin{align}
            \frac{2}{k(k-1)}\sum_{1\leq i<j\leq k}\Pr\squ*{f(x_i)=f(x_j)}\geq\alpha,
        \end{align}
        where the probability is with respect to independent evaluations of the function, that is $\Pr\squ*{f(x_i)=f(x_j)}=\sum_{y\in Y}\Pr\squ*{y=f(x_i)}\Pr\squ*{y=f(x_j)}$.

        \item A probabilistic function $f:X\rightarrow Y$ has an \emph{$(\alpha,k)$-independent set} if there exist distinct $x_1,\ldots,x_k\in X$ such that for all $i\neq j$,
        \begin{align}
            \frac{2}{k(k-1)}\sum_{1\leq i<j\leq k}\Pr\squ*{f(x_i)=f(x_j)}\leq1-\alpha.
        \end{align}
    \end{itemize}
\end{definition}

Similarly to the deterministic case, we describe functions by circuits. Here, however, to allow probabilistic functions, we will consider probabilistic circuits, where there are some bits of random input.

\begin{definition}
\begin{itemize}
    \item A \emph{probabilistic circuit} $C^m$ with $n$ inputs is a circuit $C$ on $n+m$ inputs such that on each evaluation of the circuit, the final $m$ wires are given uniformly random bits as input.
    
    \item A \emph{probabilistic Turing machine} $M^p$ with $k$ input tapes is a Turing machine $M$ with $k+1$ input tapes along with a function $p:\N\rightarrow\N$. On input $x$, $M^p(x)=M(x,r)$ where $r\in\{0,1\}^{p(|x|)}$ is sampled uniformly at random.
\end{itemize}
\end{definition}

Next, we can introduce the complexity problem for probabilistic cliques and independent sets. Rather than a language, we present it as a promise problem, which provides a probability gap between the yes and no instances.

\begin{definition}
    Let $k\in\N$ and let $c,s:\N\rightarrow(0,1)$. The \emph{probabilistic $k$-clique problem} with completeness $c$ and soundness $s$ is the promise problem $\ttt{pClique}(k)_{c,s}=(Y,N)$ with
    \begin{align}
    \begin{split}
        &Y=\set*{\text{probabilistic circuits }C^m}{f_{C^m}\text{ has a $(c(|C^m|),k)$-clique}},\\
        &N=\set*{\text{probabilistic circuits }C^m}{f_{C^m}\text{ has no $(s(|C^m|),k)$-clique}}.
    \end{split}
    \end{align}

    The \emph{probabilistic $k$-independent set problem} with completeness $c$ and soundness $s$ is the promise problem $\ttt{pIS}(k)_{c,s}=(Y,N)$ with
    \begin{align}
    \begin{split}
        &Y=\set*{\text{probabilistic circuits }C^m}{f_{C^m}\text{ has a $(c(|C^m|),k)$-independent set}},\\
        &N=\set*{\text{probabilistic circuits }C^m}{f_{C^m}\text{ has no $(s(|C^m|),k)$-independent set}}.
    \end{split}
    \end{align}
\end{definition}

Analogously to the deterministic case, the probabilistic clique and independent set problems are complete for the probabilistic analogue of $\tsf{NP}$, which is $\tsf{MA}$. First, we provide a definition of this class.

\begin{definition}
    Let $c,s:\N\rightarrow(0,1)$. A promise problem $Y,N\subseteq\{0,1\}^\ast$ is in $\tsf{MA}_{c,s}$ if there exist polynomials $p,q,l:\N\rightarrow\N$ and a probabilistic Turing machine $V^l$ with two input tapes and one output tape such that
    \begin{itemize}
        \item For all $x,y\in\{0,1\}^\ast$, $V^l$ halts on input $(x,y)$ in $q(|x|)$ steps.

        \item For all $x\in Y$, there exists $y\in\{0,1\}^{p(|x|)}$ such that $\Pr\squ*{V^l(x,y)=1}\geq c(|x|)$.

        \item For all $x\in N$ and for all $y\in\{0,1\}^\ast$, $\Pr\squ*{V^l(x,y)=1}\leq s(|x|)$.
    \end{itemize}
\end{definition}
Note that, as long as there exists $N\in\N$, a polynomial $P:\N\rightarrow\N$, and a function $E:\N\rightarrow\N$ no larger than exponential such that $c(n)-s(n)\geq\frac{1}{P(n)}$ and  $s(n),1-c(n)\geq\frac{1}{E(n)}$ for all $n\geq N$, $\tsf{MA}_{c,s}=\tsf{MA}_{\frac{2}{3},\frac{1}{3}}$, so we write this class simply $\tsf{MA}$.

Next, we show completeness of the $2$-clique and independent set problems. As in the deterministic case, we refer to \cref{sec:a-k-cliques} for the reduction from $k$-cliques and independent sets to this case.

\begin{proposition}
    There exist $c,s:\N\rightarrow(0,1)$ with constant gap such that $\ttt{pClique}(2)_{c,s}$ is $\tsf{MA}$-complete.
\end{proposition}

\begin{proof}
    First, we see that $\ttt{pClique}(2)_{c,s}\in\tsf{MA}_{c,s}$ for all $c,s$. An instance is simply a probabilistic circuit $C^m$, and the verifier expects a proof that is a $(c,2)$-clique $(x,y)$. Then, the verifier computes samples of $f_{C^m}(x)$ and $f_{C^m}(y)$, and accepts if and only if they are equal. Then, if $C^m\in Y$, $\Pr\squ*{f_{C^m}(x)=f_{C^m}(y)}\geq c(|C^m|)$, so $\Pr\squ*{V^l(C^m,(x,y))=1}\geq c(|C^m|)$. Conversely, if $C^m\in N$, $\Pr\squ*{f_{C^m}(x)=f_{C^m}(y)}\leq s(|C^m|)$, so $\Pr\squ*{V^l(C^m,(x,y))=1}\leq s(|C^m|)$, completing the proof of inclusion.
    
    Let $L=(Y,N)\in\tsf{MA}_{\frac{2}{3},\frac{1}{3}}$. Then, there exists a probabilistic polynomial time Turing machine $V^l$ such that for all $x\in Y$, there exists $y$ such that $\Pr\squ*{V^l(x,y)=1}\geq\frac{2}{3}$; and if $x\in N$, then for all $y$, $\Pr\squ*{V^l(x,y)=1}\leq\frac{1}{3}$. For each $x$, Let $C_{x}^l$ be a probabilistic circuit describing the probabilistic function $f_{C_{x}^l}(y,b)=(y,V^l(x,y)\lor b)$. If $x\in Y$, then $(y,0),(y,1)$ forms a $(2/3,2)$-clique, as
    \begin{align}
        \Pr\squ*{f_{C_{x}^l}(y,0)=f_{C_{x}^l}(y,1)}=\Pr\squ*{V^l(x,y)=1}=\frac{2}{3}.
    \end{align}
    Conversely, if $x\in N$, then suppose there exists a $(s,2)$-clique $(y,b),(y',b')$ for $s\geq 1/3$. First, since the probability of acceptance is greater than $0$, we must have $y=y'$. Also, since $(y,b)\neq(y,b')$, we can take without loss of generality $b=0$ and $b'=1$. Therefore,
    \begin{align}
        \frac{1}{3}<s\leq\Pr\squ*{f_{C_{x}^l}(y,0)=f_{C_{x}^l}(y,1)}=\Pr\squ*{V^l(x,y)=1},
    \end{align}
    which is a contradiction. This provides that $\ttt{qClique}_{\frac{2}{3},\frac{1}{3}}(2)$ is $\tsf{MA}$-hard.
\end{proof}

\begin{proposition}
    There exist $c,s:\N\rightarrow(0,1)$ with constant gap such that $\ttt{pIS}_{c,s}(2)$ is a $\tsf{MA}$-complete.
\end{proposition}

\begin{proof}
    It is, as before, direct to see that $\ttt{pIS}_{c,s}(2)\in\tsf{MA}_{c,s}$ for all $c,s$ with polynomial gap. Given a circuit $C^m$ as instance and a proof $(x,y)$, the verifier simply computes samples of $f_{C^m}(x)$ and $f_{C^m}(y)$ and accepts iff they are not equal. If $(x,y)$ is a $(c,2)$-independent set of $f_{C^m}$, then $\Pr\squ*{V^l(C^m,(x,y))=1}=\Pr\squ*{f_{C^m}(x)\neq f_{C^m}(y)}\geq c$; and conversely, if $f_{C^m}$ has no $(s,2)$-independent set, $\Pr\squ*{V^l(C^m,(x,y))=1}\leq s$.

    Now, let $L=(Y,N)\in\tsf{MA}_{\frac{2}{3},\frac{1}{3}}$ with verification circuit $V_l$, and define the probabilistic circuit $C_{x}^l$ with function $f_{C_{x}^l}(y)=(y_1\land V^l(x,y),\lnot y_1\land V^l(x,y))$. If $x\in Y$, then there exists $y$ such that $\Pr\squ*{V^l(x,y)=1}\geq\frac{2}{3}$, giving that $\Pr\squ*{f_{C_{x}^l}(y)=(y_1,\lnot y_1)}\geq\frac{2}{3}$. Taking any $z$ such that $z_1\neq y_1$, the possible values of $f_{C_{x}^l}(z)$ are $(0,0)$ and $(z_1,\lnot z_1)$, so $f_{C_{x}^l}(z)\neq(y_1,\lnot y_1)$. As such, $\Pr\squ*{f_{C_{x}^l}(y)\neq f_{C_{x}^l}(z)}\geq\Pr\squ*{f_{C_{x}^l}(y)=(y_1,\lnot y_1)}\Pr\squ*{f_{C_{x}^l}(z)\neq(y_1,\lnot y_1)}\geq\frac{2}{3}$. On the other hand, if $x\in N$, for all $y$, $\Pr\squ*{V^l(x,y)=1}\leq\frac{1}{3}$, so $\Pr\squ*{f_{C_{x}^l}(y)=(0,0)}\geq\frac{2}{3}$. Thus, for all $y,y'$, $\Pr\squ*{f_{C_{x}^l}(y)=f_{C_{x}^l}(y')}\geq \Pr\squ*{f_{C_{x}^l}(y)=(0,0)}\Pr\squ*{f_{C_{x}^l}(y')=(0,0)}\geq \frac{4}{9}$. Thus, $\ttt{pIS}_{\tfrac{2}{3},\tfrac{5}{9}}$ is $\tsf{MA}$-hard.
\end{proof}

\appendix

\section{Reductions from \textit{k}-cliques to 2-cliques}\label{sec:a-k-cliques}
	
	In this appendix, we provide reductions for deterministic, probabilistic, and quantum $k$-clique and independent set problems to the corresponding $2$-clique/independent set problem. Matching the completeness results, we are able to show the reduction for both the clique and independent set problems, but for the quantum case, we show only the reduction in the case of cliques. For conciseness, we work at the level of channels --- it is direct that the results below imply complexity-theoretic results at the level of circuits.
 
	\begin{lemma}
		Let $f:X\rightarrow Y$ be a function. Define $g:X^k\times\Z_2\rightarrow X^k\times\Z_2$ as
		\begin{align}
			g(x_1,\ldots,x_k,b)=\begin{cases}(x_1,\ldots,x_k,1)&\text{ if }f(x_1)=\ldots=f(x_k)\text{ and } x_1,\ldots,x_k\text{ all distinct}\\(x_1,\ldots,x_k,b)&\text{ else}\end{cases}\,.
		\end{align}
		Then, $g$ has a $2$-clique if and only if $f$ has a $k$-clique.
	\end{lemma}
	
	It is direct to note that the mapping from the circuit of $f$ to the circuit of $g$ is efficient, that is polynomial in $k$ and the size of the description of $f$.
	
	\begin{proof}
		Suppose $f$ has a $k$-clique. Then, there exist distinct $x_1,\ldots,x_k\in X$ such that $f(x_1)=\ldots=f(x_k)$. By definition of $g$, $g(x_1,\ldots,x_k,0)=(x_1,\ldots,x_k,1)=g(x_1,\ldots,x_k,1)$, so $g$ has a $2$-clique. Conversely, suppose $f$ does not have a $k$-clique, \emph{i.e.} the pre-image of every element in $Y$ has cardinality at most $k-1$. Then, $g$ is identity -- in particular injective -- so it does not have a $2$-clique.
	\end{proof}
	
	\begin{lemma}
		Let $f:X\rightarrow Y$ be a function. Define $g:X^k\rightarrow\Z_2$ as
		\begin{align}
			g(x_1,\ldots,x_k)=\begin{cases}1&\text{if $x_1,\ldots,x_k$ all distinct and $f(x_1),\ldots,f(x_k)$ all distinct as well}\\0&\text{else}. \end{cases}
		\end{align}
		Then, $g$ has a $2$-independent set if and only if $f$ has a $2$-independent set.
	\end{lemma}
	
	As for the previous lemma, it is direct that the construction of $g$ is efficient.
	
	\begin{proof}
		Suppose $f$ has a $k$-independent set. Then, there exist $x_1,\ldots,x_k\in X$ distinct such that $f(x_1),\ldots,f(x_k)$ are all distinct. As such $g(x_1,\ldots,x_k)=1$. But since $g(x_1,\ldots,x_1)=0$, we have a $2$-independent set $\parens*{(x_1,\ldots,x_k),(x_1,\ldots,x_1)}$. Conversely, suppose that $f$ does not have a $k$-independent set. Then, for all $x_1,\ldots,x_k\in X$ there exist $i,j$ such that $f(x_i)=f(x_j)$. As such, we always have $g(x_1,\ldots,x_k)=0$, giving that there is no $2$-independent set.
	\end{proof}
	
	Next, we extend this to the probabilistic case.
	
	\begin{lemma}
		Let $f:X\rightarrow Y$ be a probabilistic function, and let $S=\set*{(i,j)}{1\leq i<j\leq k}$. Define $g:X^k\times\Z_2\rightarrow X^k\times\Z_2^2\times S$ to be the probabilistic function where $g(x_1,\ldots,x_k,b)$ is the random variable $(x_1,\ldots,x_k,b\lor\delta_{f(x_i),f(x_j)},b\lor d,(i,j))$ for $d=1$ if and only if $x_1,\ldots,x_k$ are all distinct and $(i,j)$ sampled uniformly at random from~$S$. Then, $g$ has an $(2\alpha/k(k-1),2)$-clique if and only if $f$ has a $(\alpha,k)$-clique, with $\alpha>0$.
	\end{lemma}
	
	As in the deterministic case, this gives a reduction from the $k$-clique problem $\ttt{pClique}(k)_{c,s}$ to the $2$-clique problem $\ttt{pClique}(2)_{2c/k(k-1),2s/k(k-1)}$, which preserves the polynomial completeness-soundness gap if and only if $k$ is polynomial in the circuit size. This however is a reasonable restriction on $k$, as a proof for a clique of size larger than polynomial cannot even be read in polynomial time.
	
	\begin{proof}
		Suppose $f$ has an $(\alpha,k)$-clique. Then, there exist distinct elements $x_1,\ldots,x_k\in X$ such that ${\frac{2}{k(k-1)}\sum_{(i,j)\in S}\Pr\squ*{f(x_i)=f(x_j)}\geq\alpha}$. Then,
		\begin{align}
			\begin{split}
				\Pr\squ*{g(x_1,\ldots,x_k,0)=g(x_1,\ldots,x_k,1)}&=\sum_{(i,j),(i',j')\in S}\Pr\squ*{(i,j)=(i',j')\land \delta_{f(x_i),f(x_j)}=1}\\
				&=\frac{2}{k(k-1)}\sum_{(i,j)\leftarrow S}\Pr\squ*{f(x_i)=f(x_j)}\geq\frac{2\alpha}{k(k-1)}.
			\end{split}
		\end{align}
		Conversely, suppose that $g$ has a $(2\alpha/k(k-1),2)$-clique. Then, there exist ${x_1,\ldots,x_k,x_1',\ldots,x_k'\in X}$ and $b,b'\in\Z_2$ such that $\Pr\squ*{g(x_1,\ldots,x_k,b)=g(x_1',\ldots,x_k',b')}\geq2\alpha/k(k-1)$. If $x_i\neq x_i'$ for some $i$, then this probability is $0$, which is a contradiction. On the other hand, if $b=b'$, then the two elements of the clique are not distinct. As such, the clique must be $(x_1,\ldots,x_k,0)$, $(x_1,\ldots,x_k,1)$. Further, if $x_i=x_j$ for any $i\neq j$, then $d=0$, so $g(x_1,\ldots,x_k,0)\neq g(x_1,\ldots,x_k,1)$, giving that the $x_1,\ldots,x_k$ are all distinct. By the calculation from the previous direction, this gives that
		\begin{align}
			\alpha\leq\frac{k(k-1)}{2}\Pr\squ*{g(x_1,\ldots,x_k,b)=g(x_1',\ldots,x_k',b')}=\frac{2}{k(k-1)}\sum_{(i,j)\in S}\Pr\squ*{f(x_i)=f(x_j)},
		\end{align}
		so $x_1,\ldots,x_k$ is an $(\alpha,k)$-clique of $f$.
	\end{proof}
	
	\begin{lemma}
		Let $k\geq 3$. Let $f:X\rightarrow Y$ be a probabilistic function, and let $S=\set*{(i,j)}{1\leq i<j\leq k}$. Define $g:X^k\rightarrow\Z_2\times S$ to be the probabilistic function where $g(x_1,\ldots,x_k)$ is the random variable $(d_{i,j},(i,j))$ for $d=1$ if and only if $x_1,\ldots,x_k$ are all distinct and $f(x_i)\neq f(x_j)$ with $(i,j)$ sampled uniformly at random from~$S$. Then, $g$ has an $(1-2(1-\alpha)/k(k-1),2)$-independent set if and only if $f$ has a $(\alpha,k)$-independent set, with $\alpha>0$.
	\end{lemma}
	
	\begin{proof}
		Suppose $f$ has an $(\alpha,k)$-independent set. Then, there exist distinct $x_1,\ldots,x_k\in X$ such that $\frac{2}{k(k-1)}\sum_{(i,j)}\Pr\squ*{f(x_i)=f(x_j)}\leq 1-\alpha$. Then, I claim $(x_1,\ldots,x_k)$ and $(x_1,\ldots,x_1)$ form a $(2\alpha/k(k-1),2)$-independent set of $g$. In fact,
		\begin{align}
			\Pr\squ*{g(x_1,\ldots,x_k)=g(x_1,\ldots,x_1)}=\frac{2}{k(k-1)}\sum_{(i,j)\in S}\Pr\squ*{f(x_i)=f(x_j)}\leq\frac{2}{k(k-1)}(1-\alpha).
		\end{align}
		Conversely, suppose that $g$ has a $(1-2(1-\alpha)/k(k-1),2)$-independent set. Then, there exist $(x_1,\ldots,x_k)$ and $(x_1',\ldots,x_k')$ such that $\Pr\squ*{g(x_1,\ldots,x_k)=g(x_1',\ldots,x_k')}\leq\frac{2(1-\alpha)}{k(k-1)}$. First, this probability is $\Pr\squ*{g(x_1,\ldots,x_k)=g(x_1',\ldots,x_k')}=\Pr\squ*{g(x_1,\ldots,x_k)=0}\Pr\squ*{g(x_1',\ldots,x_k')=0}+(1-\Pr\squ*{g(x_1,\ldots,x_k)=0})(1-\Pr\squ*{g(x_1',\ldots,x_k')=0})$. If $x_1,\ldots,x_k$ are all distinct, we have
		\begin{align}
			\Pr\squ*{g(x_1,\ldots,x_k)=0}=\frac{2}{k(k-1)}\sum_{(i,j)\in S}\Pr\squ*{f(x_i)=f(x_j)},
		\end{align}
		and otherwise $\Pr\squ*{g(x_1,\ldots,x_k)=0}=1$. The same holds in the same way for $\Pr\squ*{g(x_1',\ldots,x_k')=0}$
		
		Next, consider the function $p(x,y)=xy+(1-x)(1-y)$ on $x,y\in[0,1]$. We have that, if $x\leq\frac{1}{2}$ or $y\leq\frac{1}{2}$, $p(x,y)\geq\min\{x,y\}$. Further, if $p(x,y)\leq\frac{1}{2}$, then either $x\leq\frac{1}{2}$ or $y\leq\frac{1}{2}$. The above implies that $p(\Pr\squ*{g(x_1,\ldots,x_k)=0},\Pr\squ*{g(x_1',\ldots,x_k')=0})\leq\frac{2(1-\alpha)}{k(k-1)}\leq\frac{2}{k(k-1)}\leq\frac{1}{2}$, so we may without loss of generality assume $\Pr\squ*{g(x_1,\ldots,x_k)=0}\leq\frac{1}{2},\Pr\squ*{g(x_1',\ldots,x_k')=0}$. This implies that $\sum_{(i,j)\in S}\Pr\squ*{f(x_i)=f(x_j)}\leq \frac{k(k-1)}{2}\Pr\squ*{g(x_1,\ldots,x_k)=g(x_1',\ldots,x_k')}\leq(1-\alpha)$, giving that $(x_1,\ldots,x_k)$ is an $(\alpha,k)$-independent set of $f$.
	\end{proof}

    Finally, we provide the reduction between the quantum clique problems. To do so, we make use of the techniques introduced in \cref{Sec:QMA(2)}. First, we need to add an additional channel to our toolkit, that will allow us to guarantee orthogonality of a purported clique.
	
	\begin{lemma}\label{lem:orthogonaliser}
		Let $H,K$ be Hilbert spaces and let $\Pi_H:H\otimes H\rightarrow H\otimes H$ be the projector onto the symmetric subspace. Define the set $S=\set*{(i,j)}{1\leq i<j\leq k}$ and the quantum channel $\Phi_2:\mc{B}(H^{\otimes k}\otimes K)\rightarrow\mc{B}(Q'\otimes\C^S)$ as
		\begin{align}
			\Phi_2(\rho)=\expec_{(i,j)\in S}\parens*{\Tr\squ*{\Pi_H\rho_{i,j}}\mu_Q+\Tr\squ*{(\Id-\Pi_H)\rho_{i,j}}\ketbra{\perp}}\otimes\ketbra{(i,j)},
		\end{align}
		where $\rho_{i,j}$ is the marginal on the $i$-th and $j$-th subregisters.
		Then, for any pure separable states of the form $\rho=\ketbra{\psi_1}\otimes\cdots\otimes\ketbra{\psi_k}\otimes\ketbra{\alpha}$, $\sigma=\ketbra{\psi_1}\otimes\cdots\otimes\ketbra{\psi_k}\otimes\ketbra{\beta}\in\mc{D}(H^{\otimes k}\otimes K)$, $\Tr\squ*{\Phi_2(\rho)\Phi_2(\sigma)}\leq\frac{1}{2k(k-1)}$ and if $\Tr\squ*{\Phi_2(\rho)\Phi_2(\sigma)}\geq\frac{1}{2k(k-1)}-\varepsilon$, there exist orthogonal states $\ket{\psi_1'},\ldots,\ket{\psi_k'}\in H$ such that
		\begin{align}
			\norm{\rho-\tilde{\rho}}_{\Tr}+\norm{\sigma-\tilde{\sigma}}_{\Tr}\leq4k\sqrt{k}(k-1)\sqrt{\varepsilon},
		\end{align}
		where $\tilde{\rho}=\ketbra{\psi_1'}\otimes\cdots\otimes\ketbra{\psi'_k}\otimes\ketbra{\alpha}$ and $\tilde{\sigma}=\ketbra{\psi_1'}\otimes\cdots\otimes\ketbra{\psi_k'}\otimes\ketbra{\beta}$.
	\end{lemma}
	
	\begin{proof}
		First, it is direct to see that
		\begin{align}
			\Tr\squ*{\Phi_2(\rho)\Phi_2(\sigma)}=\frac{1}{2k(k-1)}\expec_{(i,j)}\squ*{\frac{1}{2}\abs*{\braket{\psi_i}{\psi_j}}^4+\parens*{1-\abs*{\braket{\psi_i}{\psi_j}}^2}^2}\leq\frac{1}{2k(k-1)}.
		\end{align}
		Now, suppose $\Tr\squ*{\Phi_2(\rho)\Phi_2(\sigma)}\geq\frac{1}{2k(k-1)}-\varepsilon$. Let $\varepsilon_{ij}=\abs*{\braket{\psi_i}{\psi_j}}^2$. Then,
		\begin{align}
			\expec_{(i,j)}\varepsilon_{ij}\leq 2\expec_{(i,j)}1-(1-\varepsilon_{ij})^2-\frac{1}{2}\varepsilon_{i,j}^2\leq 4k(k-1)\varepsilon.
		\end{align}
		
		Let $\Psi=\begin{bmatrix}\ket{\psi_1}&\cdots&\ket{\psi_k}\end{bmatrix}$ be a $\dim H\times k$ matrix. For any matrix $\Psi'=\begin{bmatrix}\ket{\psi_1'}&\cdots&\ket{\psi_k'}\end{bmatrix}$, the Frobenius norm
		\begin{align}
			\begin{split}
				\norm{\Psi-\Psi'}_2^2&=\sum_i\norm{\ket{\psi_i}-\ket{\psi_i'}}^2\geq\frac{1}{2}\sum_i\norm{\ketbra{\psi_i}-\ketbra{\psi_i'}}_2^2\\
				&\geq\frac{1}{2k}\parens[\Big]{\sum_i\norm{\ketbra{\psi_i}-\ketbra{\psi_i'}}_2}^2.
			\end{split}
		\end{align}
		So, our goal becomes to find an isometry $\Psi'$ that minimises this norm, since this is the class of matrices whose columns are orthonormal --- this question is called the \emph{orthogonal Procrustes problem}. Let $\Psi=U\Sigma V^\dag$ be the singular-value decomposition, where $\Sigma\in\mbb{M}_{k\times k}(\C)$ is a diagonal matrix with elements $\sigma_i\geq0$ and $U\in\mbb{M}_{\dim H\times k}$, $V\in\mbb{M}_{k\times k}(\C)$ are isometries: the solution of the problem is $\Psi'=UV^\dag$. In fact, the unitary that minimises $\norm{\Psi-\Psi'}_2$ is the one that maximises $\Tr\squ*{\Psi^\dag\Psi'}=\sum_{i}\sigma_i(U^\dag\Psi' V)_{i,i}$, which by unitarity is maximised when the diagonal elements of $U^\dag\Psi' V$ are all $1$, giving $U^\dag\Psi' V=\Id$. With this candidate, it suffices to compute the norm. Note that $(\Psi^\dag\Psi)_{ij}=\braket{\psi_i}{\psi_j}$, so
		\begin{align}
			\norm{\Psi^\dag\Psi-\Id}_2^2=\sum_{i\neq j}|\braket{\psi_i}{\psi_j}|^2\leq 4k^2(k-1)^2\varepsilon.
		\end{align}
		On the other hand, $\norm{\Psi^\dag\Psi-\Id}_2^2=\norm{\Sigma^2-\Id}_2^2=\sum_i(\sigma_i^2-1)^2$. This provides that
		\begin{align}
			\norm{\Psi-\Psi'}_2^2=\norm{U\Sigma V^\dag-UV^\dag}_2^2=\sum_i(\sigma_i-1)^2\leq \sum_i(\sigma_i^2-1)^2\leq 4k^2(k-1)^2\varepsilon.
		\end{align}
		Therefore, we have orthonormal vectors $\{\ket{\psi_1'},\ldots,\ket{\psi_k'}\}$ such that $\sum_i\norm{\ketbra{\psi_i}-\ketbra{\psi_i'}}_2\leq 2\sqrt{2}k^{3/2}(k-1)\sqrt{\varepsilon}$. Taking $\tilde{\rho},\tilde{\sigma}$ as in the statement, we have immediately
		\begin{align}
			\begin{split}
				\norm{\rho-\tilde{\rho}}_{\Tr}+\norm{\sigma-\tilde{\sigma}}_{\Tr}&=\sqrt{2}\norm{\ketbra{\psi_1}\otimes\cdots\otimes\ketbra{\psi_k}-\ketbra{\psi_1'}\otimes\cdots\otimes\ketbra{\psi_k'}}_2\\
				&\leq\sqrt{2}\sum_i\norm{\ketbra{\psi_i}-\ketbra{\psi_i'}}_2\leq4k\sqrt{k}(k-1)\sqrt{\varepsilon}.
			\end{split}
		\end{align}
	\end{proof}
	
	\begin{lemma}\label{lem:k-to-2-intro}
		Let $\Phi:\mc{B}(H)\rightarrow\mc{B}(K)$ be a quantum channel and let $p_1,p_2,p_3>0$ such that $p_1+p_2+p_3=1$. Let $S=\set*{(i,j)}{1\leq i<j\leq k}$ as in the previous lemma. Let $\Pi_H$ and $\Pi_K$ be the projectors onto the symmetric subspaces of $H\otimes H$ and $K\otimes K$, respectively. Let $\Phi_1:\mc{B}(H^{\otimes k}\otimes Q)\rightarrow\mc{B}(H\otimes\C^{[k]})$ be the channel from \cref{lem:separator} and $\Phi_2:\mc{B}(H^{\otimes k}\otimes Q)\rightarrow\mc{B}(Q'\otimes\C^S)$ be the channel from \cref{lem:orthogonaliser}. Define the channel $\Phi_3:\mc{B}(H^{\otimes k}\otimes Q)\rightarrow\mc{B}(Q'\otimes\C^S)$ as 
		\begin{align}
			&\Phi_3(\rho)=\expec_{(i,j)\in S}\parens[\big]{\Tr\parens*{\Pi_K(\Phi\otimes\Phi)(\rho_{i,j})}\ketbra{\perp}+\Tr\parens*{(\Id-\Pi_K)(\Phi\otimes\Phi)(\rho_{i,j})}\mu_Q}\otimes\ketbra{(i,j)}.
		\end{align}
		Define $\Phi'=p_1\Phi_1\oplus p_2\Phi_2\oplus p_3\Phi_3$. Suppose $\Phi$ has an $(\alpha,k)$-clique. Then, $\Phi'$ has an $(\alpha',2)$-clique, with $\alpha'=\frac{p_1^2}{k}+\frac{p_2^2+p_3^2(3/2+\alpha)}{2k(k-1)}$. 
	\end{lemma}
	
	\begin{proof}
		Let $\rho_1,\ldots,\rho_k$ be an $(\alpha,k)$-clique of $\Phi$. We may, without loss of generality, suppose that $\rho_i=\ketbra{\psi_i}$ are all pure, as convex combinations cannot increase the overlap. Now, I claim that $\sigma_0=\rho_1\otimes\cdots\otimes\rho_k\otimes\ketbra{0}$ and $\sigma_1=\rho_1\otimes\cdots\otimes\rho_k\otimes\ketbra{1}$ forms a $(\alpha',2)$-clique of $\Phi'$. First, $\sigma_0$ and $\sigma_1$ are orthogonal. Next, due to \cref{lem:separator}, $\Tr\squ*{\Phi_1(\sigma_0)\Phi_1(\sigma_1)}=\frac{1}{k}$ and due \cref{lem:orthogonaliser}, $\Tr\squ*{\Phi_2(\sigma_0)\Phi_2(\sigma_1)}=\frac{1}{2k(k-1)}$. Finally, 
		\begin{align}
			\Phi_3(\sigma_y)=\expec_{(i,j)}\parens*{(\tfrac{1}{2}+\tfrac{1}{2}\Tr(\Phi(\rho_i)\Phi(\rho_j)))\ketbra{\perp}+(\tfrac{1}{2}-\tfrac{1}{2}\Tr(\Phi(\rho_i)\Phi(\rho_j)))\mu_Q}\otimes\ketbra{(i,j)}
		\end{align}
		Putting this together,
		\begin{align}
			\begin{split}
				&\Tr(\Phi'(\sigma_0)\Phi'(\sigma_1))=p_1^2\Tr(\Phi_1(\sigma_0)\Phi_1(\sigma_1))+p_2^2\Tr(\Phi_2(\sigma_0)\Phi_2(\sigma_1))+p_3^2\Tr(\Phi_3(\sigma_0)\Phi_3(\sigma_1))\\
				&=\frac{p_1^2}{k}+\frac{p_2^2}{2k(k-1)}+\frac{2p_3^2}{k(k-1)}\expec_{(i,j)}\parens*{\tfrac{1}{2}+\tfrac{1}{2}\Tr(\Phi(\rho_i)\Phi(\rho_j))}^2+\frac{1}{2}\parens*{\tfrac{1}{2}-\tfrac{1}{2}\Tr(\Phi(\rho_i)\Phi(\rho_j))}^2
			\end{split}
		\end{align}
		Since the function $(\frac{1}{2}+\frac{x}{2})^2+\frac{1}{2}(\frac{1}{2}-\frac{x}{2})^2\geq\frac{3}{8}+\frac{x}{4}$ we have that $\expec_{(i,j)}\parens*{\tfrac{1}{2}+\tfrac{1}{2}\Tr(\Phi(\rho_i)\Phi(\rho_j))}^2+\frac{1}{2}\parens*{\tfrac{1}{2}-\tfrac{1}{2}\Tr(\Phi(\rho_i)\Phi(\rho_j))}^2\geq\frac{3}{8}+\frac{\alpha}{4}$, by the clique property. Thus, we have $\Tr(\Phi'(\sigma_0)\Phi'(\sigma_1))=\frac{p_1^2}{k}+\frac{p_2^2}{2k(k-1)}+\frac{p_3^2(3/2+\alpha)}{2k(k-1)}=\alpha'$, as wanted.
	\end{proof}
	
	\begin{theorem}\label{thm:q-k-to-2-clique}
		Let $\Phi$ and $\Phi'$ be as in \cref{lem:k-to-2-intro}. Then, if $\Phi$ has no $(\beta,k)$-clique, then, for all $\eta$, there are $p_1,p_2,p_3$ such that $\Phi'$ has no $(\beta'+\eta,2)$-clique, with $\beta'=\frac{p_1^2}{k}+\frac{p_2^2+p_3^2(3/2+(5/2)\beta)}{2k(k-1)}$.
	\end{theorem}
	
	Together with the previous lemma, this gives a reduction from $\ttt{qClique}(k)_{c,s}$ to $\ttt{qClique}(2)_{c',s'}$, where as long as $c\geq\frac{5}{2}s$ with polynomial gap, $c'$ and $s'$ have polynomial gap.
	
	\begin{proof}
		The proof proceeds by two succesive applications of \cref{lem:like-5.5}. For the first application, let $A=\{(\ketbra{\psi_1}\otimes\cdots\otimes\ketbra{\psi_k}\otimes\ketbra{\alpha},\ketbra{\psi_1}\otimes\cdots\otimes\ketbra{\psi_k}\otimes\ketbra{\beta})|\ket{\psi_1},\ldots\ket{\psi_k}\in H\text{ and }\ket{\alpha},\ket{\beta}\in H'\text{ orthonormal}\}$, and $B=\set*{(\rho,\sigma)\in A}{\ket{\psi_1},\ldots\ket{\psi_k}\text{ orthonormal}}$. Then, as $\Phi$ has no $(\beta,k)$-clique, for all $(\rho,\sigma)\in B$, $\Tr\squ*{\Phi_3(\rho)\Phi_3(\sigma)}\leq \frac{3/2+(5/2)\beta}{2k(k-1)}$, due to the fact that $(\frac{1}{2}+\frac{x}{2})^2+\frac{1}{2}(\frac{1}{2}-\frac{x}{2})^2\leq\frac{3}{8}+\frac{5}{8}x$ for $x\in[0,1]$. Then, by using \cref{lem:orthogonaliser}, there exists $p\in(0,1)$ such that $p\Phi_2\oplus(1-p)\Phi_3$ has no $(\frac{p^2+(1-p)^2(3/2+(5/2)\beta)}{2k(k-1)}+\eta_1,2)$-clique. Now for the second application, we take $B$ to be $A$ from above and $A=\set*{(\ketbra{\psi},\ketbra{\phi})}$. Then, using \cref{lem:separator}, there exists $q\in(0,1)$ such that $q\Phi_1\oplus (1-q)p\Phi_2\oplus(1-q)(1-p)\Phi_3$ has no $(\frac{q^2}{k}+\frac{((1-q)p)^2+((1-q)(1-p))^2(3/2+(5/2)\beta)}{2k(k-1)}+(1-q)^2\eta_1+\eta_2,2)$-clique. Take $\eta=(1-q)^2\eta_1+\eta_2\leq\eta_1+\eta_2$, which can be made arbitrarily small, and $p_1=q$, $p_2=(1-q)p$, $p_3=(1-q)(1-p)$, which completes the proof.
	\end{proof}
    
\section{Alternate proof of \textsf{QMA}(2)-hardness}\label{a-alternateproof}

In this appendix, we give an alternate proof that the quantum clique problem is $\tsf{QMA}(2)$-hard. The proof is conceptually simpler than the proof in \cref{Sec:QMA(2)}, but provides no additional constraints on the structure of the channels.

\begin{theorem}
	There exist $c,s:\N\rightarrow(0,1)_{\exp}$ with polynomial gap such that $\ttt{qClique}(2)_{c,s}$ is $\tsf{QMA}(2)$-complete.
\end{theorem}

\begin{proof}
	Let $L\in\tsf{QMA}(2)$. Based on the construction of \cite{HM13}, we can make some assumptions about the structure of the verification circuit. First, we may assume that the witness state is pure, and that it consists of two copies of the same state $\ket{\psi}\otimes\ket{\psi}\in(\C^2)^{\otimes 2m}$. Second, we can assume that, for any classical input $x$, the verification circuit takes the form of some unitary $V_x$ acting on the space of the witness state and some polynomial-sized workspace $W=(\C^2)^{\otimes k}$, and that the verification algorithm accepts or rejects by making a measurement of the first qubit of the output of the unitary. That is, if $x\in L$, then, there exists $\ket{\psi}\in(\C^2)^{\otimes m}$ such that $\norm{(\bra{1}\otimes\Id)V_x\ket{\psi}\ket{\psi}\ket{0^k}}^2\geq c$; and if $x\notin L$, then for all $\ket{\psi},\ket{\phi}\in(\C^2)^{\otimes m}$, $\norm{(\bra{1}\otimes\Id)V_x\ket{\psi}\ket{\phi}\ket{0^k}}^2\leq s$.
	
	Now, we want to construct a channel $\Phi_x$ that has a $2$-clique if and only if $x\in L$. Write $H=(\C^2)^{\otimes m}$, $W=(C^2)^{\otimes k}$, $Y=\C^2$, and $G=(\C^2)^{\otimes 2m+k-1}$, so that $V_x:H\otimes H\otimes W\rightarrow Y\otimes G$. Run the channel $\Phi_x:\mc{B}(G\otimes\C^2)\rightarrow\mc{B}(H)$ on input $\rho$ as follows.
	\begin{enumerate}[1.]
		\item Trace out the last qubit of $\rho$ to get $\rho_G$.
		\item Act with $V_x^\dag$ on $\ketbra{1}\otimes\rho_G$.
		\item Measure $W$ in the computational basis. If the measurement result is $0^k$, do nothing to the remaining state; else replace it with a maximally-mixed state.
		\item Trace out the first copy of $H$.
	\end{enumerate}
	This construction is illustrated in \cref{fig:qma2-channel}
	
	\begin{figure}[h!]
		\centering
		\begin{tikzpicture}
			\draw (-0.5,-1) rectangle (0.5, 1) node[pos=0.5]{$V_x^\dag$};
			
			\draw (2.5,0.45) rectangle (2,-0.05) node[pos=0.5]{$\Tr$};
			\draw (2,-0.25) rectangle (1.5,-0.75) node[pos=0.5]{$M$};
			\draw (-0.25,-1.75) rectangle (0.25,-1.25) node[pos=0.5]{$\Tr$};
			
			\draw (1.75,-0.25) -- (1.75,0.8);
			\filldraw (1.65,0.7) rectangle (1.85,0.9);
			
			\draw (-1.5,0.5) node[left]{$\ket{1}$} -- (-0.5,0.5) node[above,pos=0.5]{$Y$} (0.5,0.8) -- (1.5,0.8) node[above,pos=0.5]{$H$} -- (2.75,0.8) (0.5,0.2) -- (1.5,0.2) node[above,pos=0.5]{$H$} -- (2,0.2) (-1.5,-0.5) -- (-0.5,-0.5) node[above,pos=0.5]{$G$} (0.5,-0.5) -- (1.5,-0.5) node[above,pos=0.5]{$W$} (-1.5,-1.5) -- (-0.5,-1.5) node[above,pos=0.5]{$\C^2$} -- (-0.25,-1.5);
			
			\node at (-1.7,-1) {$\Bigg\{$};
			\node at (-2,-1) {$\rho$};
			
		\end{tikzpicture}
		\caption{Construction of the channel from a $\ttt{QMA}(2)$-language verification circuit.}
		\label{fig:qma2-channel}
	\end{figure}
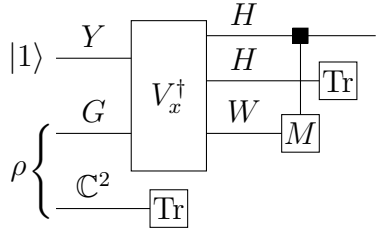
	
	First, consider the case that $x\in L$. Then, there exists $\ket{g}\in G$ such that $\abs*{\bra{1}\bra{g}V_x\ket{\psi}\ket{\psi}\ket{0^k}}^2\geq c$. I claim the wanted clique is $\rho=\ketbra{g}\otimes\ketbra{0}$, $\sigma=\ketbra{g}\otimes\ketbra{1}$. For either of those two states, the state after step 2 is $V_x^\dag\ket{1}\ket{g}=\alpha\ket{\psi}\ket{\psi}\ket{0^k}+\ket{\perp}$, where $|\alpha|^2\geq c$. Were $\alpha=1$, the output would be $\ket{\psi}$, in which case $\Tr(\Phi(\rho_0)\sigma(\sigma_0))=1$. In reality, $\norm{V_x^\dag\ket{1}\ket{g}-\ket{\psi}\ket{\psi}\ket{1}}=\norm{\ket{\perp}}=\sqrt{|1-\alpha|^2+1-|\alpha|^2}=\sqrt{2-2\latRe\alpha}\leq\sqrt{2-2c}$. Thus,
	\begin{align}
		\begin{split}
			\Tr(\Phi(\rho)\Phi(\sigma))&\geq\Tr(\Phi(\rho_0)\Pi(\sigma_0))-\norm{\Phi(\rho)-\Phi(\rho_0)}_{\Tr}-\norm{\Phi(\sigma)-\Phi(\sigma_0)}_{\Tr}\\
			&\geq 1-2\sqrt{2-2c}\geq4(c-\tfrac{7}{8})
		\end{split}
	\end{align}
	
	Now, consider the case that $x\notin L$. In this case, for all $\ket{\psi},\ket{\phi}\in H$, and $\ket{g}\in G$, $\abs*{\bra{1}\bra{g}V_x\ket{\psi}\ket{\phi}\ket{0^k}}^2\leq s$. Since every mixed state is a convex combination of pure states, if we show that for all pure states $\ket{\psi},\ket{\phi}\in G\otimes\C^2$, $\Tr(\Phi(\ketbra{\psi})\Phi(\ketbra{\phi}))$ is upper bounded by some constant, the same bound must hold for all states. Further, we may, by the same argument, work only with states that remain pure after the partial trace is taken in step 1 --- these are the separable states $\ket{\psi}=\ket{\psi_G}\otimes\ket{\psi_0}$, $\ket{\phi}=\ket{\phi_G}\otimes\ket{\phi_0}$. With this assumption, after step 2, $\ket{\psi}$ becomes $V_x^\dag\ket{1}\ket{\psi_G}=\alpha\ket{\psi_{HH}}\ket{0^k}+\beta\ket{\perp}$ for some states $\ket{\psi_{HH}}\in H\otimes H$ and $\ket{\perp}\in H\otimes H\otimes W$, where $|\alpha|^2+|\beta|^2=1$. Letting $\mu_{HH}$ be the maximally mixed state on $H\otimes H$, the state after step 3 is the $|\alpha|^2\ketbra{\psi_{HH}}+|\beta|^2\mu_{HH}$. Expanding $\ket{\psi_{HH}}=\sum_{i}\sqrt{p_i}\ket{\psi^1_i}\ket{\psi^2_i}$ via the Schmidt decomposition, the final state after step 4 is
	\begin{align}
		\Phi_x(\ketbra{\psi})=\sum_i|\alpha|^2p_i\ketbra{\psi^2_i}+|\beta|^2\mu_H=\sum_i\parens*{|\alpha|^2p_i+\tfrac{|\beta|^2}{2^m}}\ketbra{\psi^2_i}.
	\end{align}
	Note also that $|\alpha|^2p_i=\abs*{\bra{1}\bra{\psi}V_x\ket{\psi^1_i}\ket{\psi^2_i}\ket{0^k}}^2\leq s$ for all $i$. In the same way, $\Phi_x(\ketbra{\phi})=\sum_i\parens*{|\alpha'|^2q_i+\tfrac{|\beta'|^2}{2^m}}\ketbra{\phi^2_i}$ for some $|\alpha'|^2q_i\leq s$. Finally, we get that, by the Cauchy-Schwarz inequality,
	\begin{align}
		\begin{split}
			\Tr(\Phi(\ketbra{\psi})\Phi(\ketbra{\phi}))&\leq\sqrt{\sum_{i}\parens*{|\alpha|^2p_i+\tfrac{|\beta|^2}{2^m}}^2\sum_j\parens*{|\alpha'|^2q_j+\tfrac{|\beta'|^2}{2^m}}^2}\\
			&\leq\sqrt{(s+\tfrac{|\beta|^2}{2^m})\sum_{i}\parens*{|\alpha|^2p_i+\tfrac{|\beta|^2}{2^m}}(s+\tfrac{|\beta'|^2}{2^m})\sum_j\parens*{|\alpha'|^2q_j+\tfrac{|\beta'|^2}{2^m}}}\\
			&\leq s+\tfrac{1}{2^m}.
		\end{split}
	\end{align}
\end{proof}

\bibliographystyle{bibtex/bst/alphaarxiv.bst}
\bibliography{bibtex/bib/full.bib,bibtex/bib/quantum.bib,bibtex/quantum_new.bib,bibtex/complexity_refs.bib}

\end{document}